\documentclass[a4paper,10pt]{article}

\usepackage{pdfsync} 
\usepackage{microtype}
\usepackage{mathrsfs}
\usepackage{amsmath}
\usepackage{amssymb}
\usepackage{stmaryrd}
\usepackage{xspace}
\usepackage{a4wide}
\usepackage{url}
\usepackage{graphicx}

% ...

\mathcode`\<="4268  
\mathcode`\>="5269  

\mathchardef\gt="313E
\mathchardef\lt="313C

\newcommand{\termtt}{M_{\top}}
\newcommand{\termff}{M_{\bot}}
\newcommand{\termsc}{M_{succ}}
\newcommand{\termnn}[1]{V_{#1}}
\newcommand{\termxor}{M_{xor}}
\newcommand{\termxorcont}{N_{xor}}

\newcommand{\LPplus}{\Lambda_{\oplus}^{+}}

\newcommand{\rv}{\mapsto_{\mathsf{v}}}
\newcommand{\rn}{\mapsto_{\mathsf{n}}}
\newcommand{\rvn}{\mapsto}

\newcommand{\ps}[2]{{#1}\oplus{#2}}
\newcommand{\val}{\mathsf{Val}}

\newcommand{\mul}[1]{\overline{#1}}
\newcommand{\ssemv}[2]{#1\Rightarrow_\mathsf{IV} #2}
\newcommand{\cssemv}[2]{#1\Rightarrow_\mathsf{CV} #2}
\newcommand{\vssemn}[2]{#1\Rrightarrow #2}
\newcommand{\Ssemv}[1]{\mathcal{S}_\mathsf{IV}(#1)}
\newcommand{\CSsemv}[1]{\mathcal{S}_\mathsf{CV}(#1)}
\newcommand{\Ssemvs}[1]{{#1}^{\mathcal{S}}_\mathsf{IV}}
\newcommand{\Ssemns}[1]{{#1}^{\mathcal{S}}_\mathsf{IN}}
\newcommand{\ibsemv}[2]{#1\Downarrow_{\mathsf{IV}} #2}
\newcommand{\bsemv}[2]{#1\Downarrow_{\mathsf{CV}} #2}
\newcommand{\Bsemv}[1]{\mathcal{B}_\mathsf{CV}(#1)}
\newcommand{\IBsemv}[1]{\mathcal{B}_\mathsf{IV}(#1)}

\newcommand{\Bsemvs}[1]{{#1}^{\mathcal{B}}_{\mathsf{CV}}}
\newcommand{\Bsemns}[1]{{#1}^{\mathcal{B}}_{\mathsf{CN}}}
\newcommand{\sdsemv}[2]{#1\Rightarrow_{#2}^\infty}
\newcommand{\Sdsemv}[1]{\mathcal{D}_\mathsf{IV}(#1)}
\newcommand{\ssemn}[2]{#1\Rightarrow_\mathsf{IN} #2}
\newcommand{\Ssemn}[1]{\mathcal{S}_\mathsf{IN}(#1)}
\newcommand{\bsemn}[2]{#1\downarrow_\mathsf{CN} #2}
\newcommand{\Bsemn}[1]{\mathcal{B}_\mathsf{CN}(#1)}
\newcommand{\sdsemn}[2]{#1\Rightarrow_{#2}^\infty}
\newcommand{\Sdsemn}[1]{\mathcal{D}_n(#1)}

\newcommand{\Ssemvu}{\mathcal{S}_\mathsf{IV}}
\newcommand{\CSsemvu}{\mathcal{S}_\mathsf{CV}}
\newcommand{\IBsemvu}{\mathcal{B}_\mathsf{IV}}
\newcommand{\Bsemvu}{\mathcal{B}_\mathsf{CV}}

\newcommand{\sumd}[1]{\sum{#1}}
\newcommand{\interv}[2]{[#1,#2]}

\newcommand{\enc}[1]{\ulcorner #1\urcorner}

\newcommand{\emdist}{\emptyset}

\newcommand{\RRp}[2]{\RR_{[#1,#2]}}

\newcommand{\idt}{\lambda x.x}

\newcommand{\mapvn}[1]{\llceil #1\rrceil }
\newcommand{\mapnv}[1]{\llfloor #1\rrfloor}

\newcommand{\sev}{\mathsf{se_\mathsf{IV}}}
\newcommand{\svv}{\mathsf{sv_\mathsf{IV}}}
\newcommand{\smv}{\mathsf{sm_\mathsf{IV}}}
\newcommand{\csvv}{\mathsf{sv_\mathsf{CV}}}
\newcommand{\csmv}{\mathsf{sm_\mathsf{CV}}}
\newcommand{\dvv}{\mathsf{dv_v}}
\newcommand{\dmv}{\mathsf{dm_v}}
\newcommand{\bvv}{\mathsf{bv_v}}
\newcommand{\bav}{\mathsf{ba_v}}
\newcommand{\bsv}{\mathsf{bs_v}}
\newcommand{\ibev}{\mathsf{bv_\mathsf{IV}}}
\newcommand{\ibvv}{\mathsf{bv_\mathsf{IV}}}
\newcommand{\ibav}{\mathsf{ba_\mathsf{IV}}}
\newcommand{\ibsv}{\mathsf{bs_\mathsf{IV}}}

\newcommand{\sen}{\mathsf{se_n}}
\newcommand{\svn}{\mathsf{sv_n}}
\newcommand{\smn}{\mathsf{sm_n}}
\newcommand{\dvn}{\mathsf{dv_n}}
\newcommand{\dmn}{\mathsf{dm_n}}
\newcommand{\bvn}{\mathsf{bv_n}}
\newcommand{\ban}{\mathsf{ba_n}}
\newcommand{\bsn}{\mathsf{bs_n}}

\newcommand{\NN}{\mathbb{N}}

\newcommand{\RR}{\mathbb{R}}

\newcommand{\PCF}{\ensuremath{\mathsf{PCF}}}
\newcommand{\PLCF}{\ensuremath{\mathsf{PLCF}}}
\newcommand{\LCF}{\ensuremath{\mathsf{LCF}}}
\newcommand{\LOP}{\ensuremath{\Lambda_{\oplus}}}
\newcommand{\LOPplus}{\ensuremath{\Lambda_{\oplus}}^{+}}

\newcommand{\zrule}[2]{%
  \prooftree \justifies #1 \using #2 \endprooftree}
\newcommand{\urule}[3]{%
  \prooftree #1 \justifies #2 \using #3 \endprooftree}
\newcommand{\brule}[4]{%
  \prooftree #1\ \ \ #2 \justifies #3 \using #4 \endprooftree}
\newcommand{\trule}[5]{%
  \prooftree #1\ \ \ #2 \ \ #3\justifies #4 \using #5 \endprooftree}

\newcommand{\czrule}[2]{%
  \prooftree \Justifies #1 \using #2 \endprooftree}

\newcommand{\cbrule}[4]{%
  \prooftree #1\ \ \ #2 \Justifies #3 \using #4 \endprooftree}
\newcommand{\ctrule}[5]{%
  \prooftree #1\ \ \ #2 \ \ #3\Justifies #4 \using #5 \endprooftree}

\newcommand{\svarone}{\alpha}
\newcommand{\svartwo}{\beta}
\newcommand{\svarthree}{\gamma}

\newcommand{\svarfive}{\epsilon}
\newcommand{\termone}{M}
\newcommand{\termtwo}{N}
\newcommand{\termthree}{L}
\newcommand{\termfour}{P}
\newcommand{\termfive}{Q}
\newcommand{\termsix}{R}
\newcommand{\termseven}{S}

\newcommand{\varone}{x}
\newcommand{\vartwo}{y}
\newcommand{\varthree}{z}
\newcommand{\varfour}{w}
\newcommand{\varfive}{s}
\newcommand{\varsix}{t}
\newcommand{\valone}{V}
\newcommand{\valtwo}{W}
\newcommand{\valthree}{X}

\newcommand{\valfive}{Z}
\newcommand{\ctone}{K}
\newcommand{\cttwo}{J}
\newcommand{\distone}{\mathscr{D}}
\newcommand{\disttwo}{\mathscr{E}}
\newcommand{\distthree}{\mathscr{F}}
\newcommand{\distfour}{\mathscr{G}}
\newcommand{\distfive}{\mathscr{H}}
\newcommand{\distsix}{\mathscr{I}}
\newcommand{\distseven}{\mathscr{J}}
\newcommand{\disteight}{\mathscr{K}}
\newcommand{\pdone}{\mathcal{D}}
\newcommand{\pdtwo}{\mathcal{E}}
\newcommand{\pdthree}{\mathcal{F}}
\newcommand{\setone}{Q}
\newcommand{\probone}{p}

\newcommand{\derone}{\pi}
\newcommand{\dertwo}{\rho}
\newcommand{\derthree}{\xi}
\newcommand{\derfour}{\mu}
\newcommand{\derfive}{\phi}
\newcommand{\unone}{\mathcal{U}}

\newcommand{\juone}{c}
\newcommand{\sjone}{A}
\newcommand{\sjtwo}{B}
\newcommand{\isone}{\Phi}
\newcommand{\funone}{f}

\newcommand{\elone}{a}
\newcommand{\natone}{n}

\newcommand{\strone}{s}

\newcommand{\adm}[2]{#1:#2}

\newcommand{\ide}{I}

\newcommand{\appr}[1]{\mathit{Appr}_{#1}}

\newcommand{\pd}[1]{\{#1\}}

\newcommand{\subst}[3]{#1\{#2/#3\}}

\newcommand{\ltder}{\prec}
\newcommand{\leqder}{\preceq}

\newcommand{\setvar}{\mathsf{X}}

\newcommand{\setvarg}{\mathsf{C}}

\newcommand{\abstr}[2]{\lambda #1.#2}
\newcommand{\app}[2]{#1#2}
\newcommand{\pair}[2]{\langle #1,#2\rangle}

\newcommand{\lfp}[1]{\mathit{lfp}(#1)}
\newcommand{\gfp}[1]{\mathit{gfp}(#1)}

\newcommand{\opinf}[1]{F_{#1}}

\newcommand{\supp}[1]{\mathsf{S}(#1)}

\newcommand{\termsplit}{M_\mathit{split}}
\newcommand{\termfdt}{M_\mathit{fdt}}
\newcommand{\termfpc}{H}

\newenvironment{varitemize}
{
\begin{list}{\labelitemi}
{\setlength{\itemsep}{0pt}
 \setlength{\topsep}{0pt}
 \setlength{\parsep}{0pt}
 \setlength{\partopsep}{0pt}
 \setlength{\leftmargin}{15pt}
 \setlength{\rightmargin}{0pt}
 \setlength{\itemindent}{0pt}
 \setlength{\labelsep}{5pt}
 \setlength{\labelwidth}{10pt}
}}
{
 \end{list} 
}

\newcounter{numberone}

\newenvironment{varenumerate}
{
\begin{list}{\arabic{numberone}.}
{
  \usecounter{numberone}
  \setlength{\itemsep}{0pt}
  \setlength{\topsep}{0pt}
  \setlength{\parsep}{0pt}
  \setlength{\partopsep}{0pt}
  \setlength{\leftmargin}{15pt}
  \setlength{\rightmargin}{0pt}
  \setlength{\itemindent}{0pt}
  \setlength{\labelsep}{5pt}
  \setlength{\labelwidth}{15pt}
}}
{
\end{list} 
}

\newcounter{numbertwo}

\newcommand{\FV}[1]{\mathsf{FV}(#1)}

\newtheorem{theorem}{Theorem}
\newtheorem{proposition}{Proposition}
\newtheorem{lemma}{Lemma}
\newtheorem{definition}{Definition}
\newtheorem{fact}{Fact}
\newtheorem{example}{Example}
\newtheorem{corollary}{Corollary}

\newtheorem{notation}{Notation}
\newenvironment{proof}{\begin{trivlist}
       \item[\hskip \labelsep {\bfseries Proof.}]}{\hfill $\Box$ \end{trivlist}}

\newcommand{\vssemv}[2]{#1\twoheadrightarrow #2}

\newcommand{\parkl}{\mathscr{L}_{\gamma}}

\newcommand{\df}{\stackrel{\mbox{\textsf{def}}}{=}}

\newcommand{\pdists}{\mathcal{P}}

\usepackage{amsmath}

\usepackage{xy}
\xyoption{frame}
\xyoption{arrow}
\xyoption{curve}
\xyoption{arc}
\xyoption{line}
\xyoption{matrix}
\xyoption{tile}
\xyoption{ps}

\newcount\mscount
\newdimen\proofrulebreadth \proofrulebreadth=.05em
\newdimen\proofdotseparation \proofdotseparation=1.25ex
\newdimen\proofrulebaseline \proofrulebaseline=2ex
\newcount\proofdotnumber \proofdotnumber=3
\let\then\relax
\def\hfi{\hskip0pt plus.0001fil}
\mathchardef\squigto="3A3B
\newif\ifinsideprooftree\insideprooftreefalse
\newif\ifonleftofproofrule\onleftofproofrulefalse
\newif\ifproofdots\proofdotsfalse
\newif\ifdoubleproof\doubleprooffalse
\let\wereinproofbit\relax
\newdimen\shortenproofleft
\newdimen\shortenproofright
\newdimen\proofbelowshift
\newbox\proofabove
\newbox\proofbelow
\newbox\proofrulename
\def\shiftproofbelow{\let\next\relax\afterassignment\setshiftproofbelow\dimen0 }
\def\shiftproofbelowneg{\def\next{\multiply\dimen0 by-1 }%
\afterassignment\setshiftproofbelow\dimen0 }
\def\setshiftproofbelow{\next\proofbelowshift=\dimen0 }
\def\setproofrulebreadth{\proofrulebreadth}

\def\prooftree{
\ifnum  \lastpenalty=1
\then   \unpenalty
\else   \onleftofproofrulefalse
\fi
\ifonleftofproofrule
\else   \ifinsideprooftree
        \then   \hskip.5em plus1fil
        \fi
\fi
\bgroup
\setbox\proofbelow=\hbox{}\setbox\proofrulename=\hbox{}%
\let\justifies\proofover\let\leadsto\proofoverdots\let\Justifies\proofoverdbl
\let\using\proofusing\let\[\prooftree
\ifinsideprooftree\let\]\endprooftree\fi
\proofdotsfalse\doubleprooffalse
\let\thickness\setproofrulebreadth
\let\shiftright\shiftproofbelow \let\shift\shiftproofbelow
\let\shiftleft\shiftproofbelowneg
\let\ifwasinsideprooftree\ifinsideprooftree
\insideprooftreetrue
\setbox\proofabove=\hbox\bgroup$\displaystyle 
\let\wereinproofbit\prooftree
\shortenproofleft=0pt \shortenproofright=0pt \proofbelowshift=0pt
\onleftofproofruletrue\penalty1
}

\def\eproofbit{
\ifx    \wereinproofbit\prooftree
\then   \ifcase \lastpenalty
        \then   \shortenproofright=0pt  
        \or     \unpenalty\hfil         
        \or     \unpenalty\unskip       
        \else   \shortenproofright=0pt  
        \fi
\fi
\global\dimen0=\shortenproofleft
\global\dimen1=\shortenproofright
\global\dimen2=\proofrulebreadth
\global\dimen3=\proofbelowshift
\global\dimen4=\proofdotseparation
\global\count255=\proofdotnumber
$\egroup  
\shortenproofleft=\dimen0
\shortenproofright=\dimen1
\proofrulebreadth=\dimen2
\proofbelowshift=\dimen3
\proofdotseparation=\dimen4
\proofdotnumber=\count255
}

\def\proofover{
\eproofbit 
\setbox\proofbelow=\hbox\bgroup 
\let\wereinproofbit\proofover
$\displaystyle
}%
\def\proofoverdbl{
\eproofbit 
\doubleprooftrue
\setbox\proofbelow=\hbox\bgroup 
\let\wereinproofbit\proofoverdbl
$\displaystyle
}%
\def\proofoverdots{
\eproofbit 
\proofdotstrue
\setbox\proofbelow=\hbox\bgroup 
\let\wereinproofbit\proofoverdots
$\displaystyle
}%
\def\proofusing{
\eproofbit 
\setbox\proofrulename=\hbox\bgroup 
\let\wereinproofbit\proofusing
\kern0.3em$
}

\def\endprooftree{
\eproofbit 
  \dimen5 =0pt
\dimen0=\wd\proofabove \advance\dimen0-\shortenproofleft
\advance\dimen0-\shortenproofright
\dimen1=.5\dimen0 \advance\dimen1-.5\wd\proofbelow
\dimen4=\dimen1
\advance\dimen1\proofbelowshift \advance\dimen4-\proofbelowshift
\ifdim  \dimen1<0pt
\then   \advance\shortenproofleft\dimen1
        \advance\dimen0-\dimen1
        \dimen1=0pt
        \ifdim  \shortenproofleft<0pt
        \then   \setbox\proofabove=\hbox{%
                        \kern-\shortenproofleft\unhbox\proofabove}%
                \shortenproofleft=0pt
        \fi
\fi
\ifdim  \dimen4<0pt
\then   \advance\shortenproofright\dimen4
        \advance\dimen0-\dimen4
        \dimen4=0pt
\fi
\ifdim  \shortenproofright<\wd\proofrulename
\then   \shortenproofright=\wd\proofrulename
\fi
\dimen2=\shortenproofleft \advance\dimen2 by\dimen1
\dimen3=\shortenproofright\advance\dimen3 by\dimen4
\ifproofdots
\then
        \dimen6=\shortenproofleft \advance\dimen6 .5\dimen0
        \setbox1=\vbox to\proofdotseparation{\vss\hbox{$\cdot$}\vss}%
        \setbox0=\hbox{%
                \advance\dimen6-.5\wd1
                \kern\dimen6
                $\vcenter to\proofdotnumber\proofdotseparation
                        {\leaders\box1\vfill}$%
                \unhbox\proofrulename}%
\else   \dimen6=\fontdimen22\the\textfont2 
        \dimen7=\dimen6
        \advance\dimen6by.5\proofrulebreadth
        \advance\dimen7by-.5\proofrulebreadth
        \setbox0=\hbox{%
                \kern\shortenproofleft
                \ifdoubleproof
                \then   \hbox to\dimen0{%
                        $\mathsurround0pt\mathord=\mkern-6mu%
                        \cleaders\hbox{$\mkern-2mu=\mkern-2mu$}\hfill
                        \mkern-6mu\mathord=$}%
                \else   \vrule height\dimen6 depth-\dimen7 width\dimen0
                \fi
                \unhbox\proofrulename}%
        \ht0=\dimen6 \dp0=-\dimen7
\fi
\let\doll\relax
\ifwasinsideprooftree
\then   \let\VBOX\vbox
\else   \ifmmode\else$\let\doll=$\fi
        \let\VBOX\vcenter
\fi
\VBOX   {\baselineskip\proofrulebaseline \lineskip.2ex
        \expandafter\lineskiplimit\ifproofdots0ex\else-0.6ex\fi
        \hbox   spread\dimen5   {\hfi\unhbox\proofabove\hfi}%
        \hbox{\box0}%
        \hbox   {\kern\dimen2 \box\proofbelow}}\doll%
\global\dimen2=\dimen2
\global\dimen3=\dimen3
\egroup 
\ifonleftofproofrule
\then   \shortenproofleft=\dimen2
\fi
\shortenproofright=\dimen3
\onleftofproofrulefalse
\ifinsideprooftree
\then   \hskip.5em plus 1fil \penalty2
\fi
}

\urldef\mailugo\url{dallago@cs.unibo.it}
\urldef\mailmarghi\url{margherita.zorzi@univr.it}
\title{Probabilistic Operational Semantics\\ for the Lambda Calculus}
\author{Ugo Dal Lago\footnote{Universit\`a di Bologna \& EPI FOCUS, \mailugo}
        \and 
        Margherita Zorzi\footnote{Universit\`a di Verona, \mailmarghi}}
\date{}

\begin{document}

\maketitle

\begin{abstract}  
Probabilistic operational semantics for a nondeterministic extension of pure lambda calculus is studied. 
In this semantics, a term evaluates to a (finite or infinite) distribution of 
values. Small-step and big-step semantics are both inductively and coinductively defined. Moreover, small-step and big-step semantics are shown to produce identical outcomes, both in call-by-value and in call-by-name. 
Plotkin's CPS translation is extended to accommodate the choice operator and shown correct with respect to
the operational semantics. Finally, the expressive power of the obtained system is studied: the calculus  
is shown to be sound and complete with respect to computable probability distributions.
\end{abstract}

\section{Introduction}\label{sec:intro}
Randomized computation is central to several areas of theoretical computer science, including
computational complexity, cryptography, analysis of 
computation dealing with uncertainty, incomplete knowledge agent systems. 
Some works have been devoted also to the design and analysis of programming languages with stochastic aspects.
For various reasons, the functional programming
paradigm seems appropriate in this context, because of the very thin gap
between the realm of programs and the underlying probability world.

The large majority of the literature on probabilistic functional programming view
probability as a \emph{monad}, in the sense of Moggi~\cite{Moggi89bis,Moggi89}.
This is the case for the works by Plotkin and Jones about denotational
semantics of probabilistic functional programs~\cite{JonesPlotkin89}, or for many of the
recently proposed probabilistic functional programming languages: 
the stochastic lambda calculus~\cite{RamseyPfeffer02} 
and the lambda calculi by Park, Pfenning and Thrun~\cite{ParkPfenningThrun03,ParkPfenningThrun05}.
The monadic structure of probability distributions provides a good denotational model for the calculi and it makes
evident how the mathematical foundations of probability can be applied directly in a natural way to the semantics of 
probabilistic programs. This allows, for example, to apply this approach to the formalization 
of properties of randomized algorithms in interactive theorem proving~\cite{PaulinAudebaud}.
The monadic approach seems particularly appropriate in applications, since some programming
languages, like Haskell, directly implement monads.

But there is another, more direct, way to endow the lambda calculus with probabilistic
choice, namely by enriching it with a binary choice operator $\oplus$. This way, we
can form terms whose behavior is probabilistically chosen as the one of the first or of
the second argument. It is not clear, however, whether the operational theory 
underlying ordinary, deterministic lambda calculus, extends smoothly to this new, probabilistic setting. 
The aim of this paper is precisely to start an investigation in this direction. The \emph{object}
of our study will be the nondeterministic lambda calculus, being it a minimal extension
of the ordinary lambda calculus with a choice operator. The \emph{subject} of our study, on the other hand, 
will be the properties of two notions of probabilistic semantics for it, namely 
call-by-value and call-by-name evaluation.
Big-step and small-step probabilistic semantics will be defined and proved equivalent, both
when defined inductively and when defined coinductively.
CPS translations extending the ones in the literature are presented
and proved to have the usual properties enjoyed in the deterministic case. Finally, some 
results about the expressive power of the obtained calculus are proved.

\subsection{Related works}\label{sec:relw}
The pioneering  investigation in the fields of stochastic functional languages is the probabilistic 
\LCF\ (\PLCF\ in the following) proposed by Saheb-Djahromi in~\cite{Saheb78}. \PLCF\ is a typed, 
higher-order calculus based on Milner's \LCF, Plotkin's \PCF~\cite{Plotkin77} and Plotkin's 
probabilistic domains (further developed by Plotkin and Jones in ~\cite{JonesPlotkin89}). The syntax of  
\PLCF, two kinds of abstractions are present which deal separately with call-by-value and call-by-name evaluation, 
respectively. The author declares an explicit intent in providing ``a foundation for the probabilistic study of the computation'' 
and even if a number of  important aspects are unexplored, the approach is interesting and related to the present investigation. 
Saheb-Djahromi provides both denotational and operational semantics for \PLCF.
Denotational semantics, defined in terms of probabilistic domains, is an extension of Milner's and Plotkin's one.
Operational semantics is a given as a Markov chain, and
an equivalence result between the latter and a denotational model is stated and proved 
as an extension of Plotkin's results.

In recent years, some lambda calculi with probabilistic features have been introduced,
strongly oriented to applications (e.g. robotics). The most developed approach is definitely 
the monadic one~\cite{Moggi89bis,Moggi89}, based on the idea that probability distribution forms a 
monad (\cite{Giry81},\cite{JonesPlotkin89}).
In~\cite{RamseyPfeffer02}, Ramsey and Pfeffer introduce the \emph{stochastic lambda calculus}, 
in which the denotation of expressions are distributions and in which the probability monad is
shown to be useful in the evaluation of queries about the defined probabilistic model. A measure 
theory and a simple measure language are introduced and the stochastic lambda calculus is
compiled into the measure language. Moreover, a denotational semantics based upon the monad of probability measures is defined.
In~\cite{Park03}, Park define the typed calculus $\lambda_{\gamma}$, an extension of the typed lambda calculus with an 
\emph{integral abstraction} and a \emph{sampling construct}, which bind probability and sampling variables respectively. 
A system of simple types for $\lambda_{\gamma}$ is also introduced. The author briefly discusses
about  the expressive power of $\lambda_{\gamma}$: the calculus has been shown to be able to express 
the most relevant probability distributions. $\lambda_{\gamma}$ does not make use of monads, which are however
present in~\cite{ParkPfenningThrun03}, in which the idea of a calculus based on the 
mathematical notion of a sampling function is further developed through the introduction of
$\lambda_{\circ}$. $\lambda_{\circ}$ is based on the monad of sampling functions and is able 
to specify probability distributions over infinite discrete domains and continuous domains. The authors also develop
a new operational semantics, called \emph{horizontal operational semantics}. The calculus $\lambda_{\circ}$ is
further studied and developed in~\cite{ParkPfenningThrun05}.

Nondeterminism and probability, however, can find their place in a completely different way.
In~\cite{deLigPip95}, de'Liguoro and Piperno propose a non deterministic lambda calculus, called \LOP. 
\LOP\ is nothing more than the usual, untyped, $\lambda$-calculus with an additional binary operator $\oplus$
which serves to represent binary, nondeterministic choice: $\termone\oplus\termtwo$ rewrites to either
$\termone$ or $\termtwo$. The authors give a standardization theorem and an algebraic semantics for \LOP. 
The classical definition of a B\"ohm tree is  extended to the non-deterministic case by means of inductively
defined ``approximating operators''.
Several relevant properties such as discriminability are studied and, 
moreover, some suitable models for non-deterministic lambda calculus are proposed and discussed.
In~\cite{Ale05}, Di Pierro, Henkin and Wiklicky propose an untyped lambda calculus with probabilistic choice. 
Its syntax is itself an extension of pure lambda calculus with $n$-ary probabilistic choice in the form
$\bigoplus_{i=1}^np_i:\termone_i$. The main objective, however, is showing how probabilistic abstract interpretation can be exploited
in the context of static analyis of probabilistic programs even in presence of higher-order functions.
This is reflected by its operational semantics, which is more directed to program analysis (when two terms
can be considered as equal?) than to computation (what is the value obtained by evaluating a program? how can
we compute it?).

\subsection{Outline}
After some motivating observations about the interplay between rewriting and nondeterministic choice (Section~\ref{sec:obs}),
the calculus $\LOP$ and its call-by-value probabilistic operational
semantics are introduced (Section~\ref{sec:syntax} and Section~\ref{sec:cbv}, respectively). Both small-step and big-step
semantics are defined and their equivalence is proved in detail. Remarkably, the result holds even when operational
semantics is formulated coinductively. The same results hold for call-by-name evaluation and are described
in Section~\ref{sec:cbn}. Call-by-value and call-by-name can be shown to be able to simulate each other by slight modifications
of the well-known CPS translations~\cite{Plotkin75}, described in Section ~\ref{sec:trans}. The paper ends with a result about
the expressive power of the calculus namely the equivalence between representable and computable distributions. 

\section{Some Motivating Observations}\label{sec:obs}
Lambda calculus can be seen both as an equational theory on lambda terms and as an abstract
model of computation. In the first case, it is natural (and customary) to allow to apply
equations (e.g. $\beta$ or $\eta$ equivalences) at any position in the term. The obtained
calculus enjoys confluence, in the form of the the so-called Church-Rosser theorem: equality
of terms remains transitive even if equations becomes rewriting rules, i.e. if they are given with an orientation.
More computationally, on the other hand, the meaning of any lambda-term is the \emph{value}
it evaluates to in some strategy or machine. In this setting,
abstractions are often considered as values, meaning that
reduction cannot take place in the scope of a lambda abstraction. What's obtained this way is a calculus
with \emph{weak} reduction which is not confluent in the Church-Rosser sense. As an example, take the term 
$\app{(\abstr{\varone}{\abstr{\vartwo}{\varone}})}{(\app{(\abstr{\varthree}{\varthree})}{(\abstr{\varthree}{\varthree})})}$.
In call-by-value, it reduces to ${\abstr{\vartwo}{(\abstr{\varthree}{\varthree})}}$. In call-by-name,
it reduces to ${\abstr{\vartwo}{\app{(\abstr{\varthree}{\varthree})}{(\abstr{\varthree}{\varthree})}}}$.
A beautiful operational theory has developed since Plotkin's pioneering's work~\cite{Plotkin75}. Call-by-value
and call-by-name have been shown to be dual to each other~\cite{Curien00theduality}, and continuation-based translations allowing to
simulate one style with the other have been designed and analyzed very carefully~\cite{DanvyFil92, DanvyNielsen05}.

Now, suppose to endow the lambda calculus with nondeterministic sums. Suppose, in other words, to introduce a binary infix
term operator $\oplus$ such that $\ps{\termone}{\termtwo}$ can act either as $\termone$ or
$\termtwo$ in a nondeterministic flavor. What we obtain is the so-called nondeterministic
lambda-calculus, which has been introduced and studied in~\cite{deLigPip95}. We cannot hope to get
any confluence results in this setting (at least if we stick to reduction as a binary relation on \emph{terms}): 
a term like $\termone=\ps{\termtt}{\termff}$.
(where $\termtt=\abstr{\varone}{\abstr{\vartwo}{\varone}}$ and $\termff=\abstr{\varone}{\abstr{\vartwo}{\vartwo}}$ 
are the usual representation of truth values in the lambda-calculus)
reduces to two distinct values (which are different in a very strong sense!) in any strategy.
The meaning of any lambda term, here, is a \emph{set of values} accounting for all the possible outcomes.
The meaning of $\termone$, as an example, is the set $\{\termtt,\termff\}$.
Nontermination, mixed with nondeterministic choice, implies
that the meaning of terms can even be an \emph{infinite} set of values. As an example, take the
lambda term
$$
\app{(\app{Y}{(\abstr{\varone}{\abstr{\vartwo}{(\ps{\vartwo}{\app{\varone}{(\app{\termsc}{\vartwo})}})}})})}{\termnn{0}}
$$
where $Y$ is a fixed-point combinator, $\termnn{\natone}$ represents the natural number $\natone\in\NN$ and
$\termsc$ computes the successor. Evaluating it (e.g. in call-by-value) produces the infinite set 
$$
\{\termnn{0},\termnn{1},\termnn{2},\ldots\}.
$$ 
It's clear that ordinary ways to give an operational semantics
to the lambda-calculus (e.g. a finitary, inductively defined, formal system) do not suffice here, since
they intrinsically attribute a finitary meaning to terms. So, how can we define small-step and
big-step semantics in a nondeterministic setting?

Another problematic point is confluence. The situation is even worse
than in ordinary, deterministic, lambda calculus. Take the following term~\cite{Saheb78,SelVal06}:
$$
\termone=\app{(\abstr{ \varone} {\app{\app{\termxor}{\varone}}{\varone}})}{(\ps{\termtt}{\termff})}
$$
where $\termxor=\abstr{\varone}{\abstr{\vartwo}{(\varone(\abstr{\varthree}{\varthree\termff\termtt})(\abstr{\varthree}{\varthree\termtt\termff}))\vartwo}}$
is a term computing a parity function of the two bits in input.
When reducing it call-by-value, we obtain the outcome $\{\termff\}$, while reducing it call-by-name, we
obtain $\{\termtt,\termff\}$. This phenomenon is due to the interaction between nondeterministic choice and 
copying: in call-by-value we choose before copying, and the final result can only be of a certain
form. In call-by-name, we copy before choosing, getting distinct outcomes.
What happens to CPS translations in this setting? Is it still possible to define them?

The aim of this paper is precisely to give answers to the questions above. Or, better, to give
answers to their natural, quantitative generalizations obtained by considering $\oplus$ as an operator
producing any of two possible outcomes with identical probability.

\section{A Brief Introduction to Coinduction}\label{sec:coind}
\emph{Coinduction} is a definitional principle dual to induction, which support a proof principle. It is a very 
useful instrument for reasoning about diverging computations and unbounded structures. 
Coinductive techniques are not yet as popular as inductive techniques; for this reason, in this section we give 
a short introduction to coinduction, following \cite{Leroy}.

An inference system over a set $\unone$ of \emph{judgments} is a set of inference rules. 
An \emph{inference rule} is a pair $(\sjone,\juone)$, where 
$\juone\in\unone$ is the conclusion of the rule and $\sjone\subseteq\unone$ is the set of its premises or 
antecedents. An \emph{inference system} $\isone$ over $\unone$ is a set of inference rules over $\unone$.

The usual way to give meaning to an inference system is to consider the fixed points 
of the associated inference operator. If $\isone$ 
is an inference system over $\unone$, we define the operator 
$\opinf{\isone}:\wp(\unone)\rightarrow\wp(\unone)$ as 
$$
\opinf{\isone}(\sjone)=\{\juone\in\unone\;|\;\exists\sjtwo\subseteq\sjone,\;(\sjtwo,\juone)\in\isone\}. 
$$
In other words, $\opinf{\isone}(\sjone)$ is the set of judgments that can be inferred in \emph{one step} from the 
judgments in $\sjone$ by using the inference rules. 
A set $\sjone$ is said to be \emph{closed} if $\opinf{\isone}(\sjone)\subseteq\sjone$, and \emph{consistent} if
 $\sjone\subseteq\opinf{\isone}{\sjone}$.
A closed set $\sjone$ is such that no new judgments can be inferred from $\sjone$. A consistent set $\sjone$ is such that all 
judgments that cannot be inferred from $\sjone$ are not in $\sjone$. The inference operator is monotone: $\opinf{\isone}(\sjone)
\subseteq\opinf{\isone}(\sjtwo)$ if $\sjone\subseteq\sjtwo$. By Tarski's fixed point theorem for complete lattices, it 
follows that the inference operator possesses both a least fixed point and a greatest fixed point, which are the 
smallest closed set and the largest consistent set, respectively: 
\begin{align*}
\lfp{\opinf{\isone}}&=\bigcap\{\sjone\;|\;\opinf{\isone}(\sjone)\subseteq\sjone\};\\ 
\gfp{\opinf{\isone}}&=\bigcup\{\sjone\;|\;\sjone\subseteq\opinf{\isone}(\sjone)\}. 
\end{align*}
The least fixed point $\lfp{\opinf{\isone}}$ is the \emph{inductive interpretation} of the inference system $\isone$, and 
the greatest fixed point $\gfp{\opinf{\isone}}$ is its \emph{coinductive interpretation}. 

These interpretations lead to the 
following two proof principles: 
\begin{varitemize}
\item
  \textbf{Induction principle}: to prove that all judgments in the inductive interpretation belong to a set $\sjone$, 
  show that $\sjone$ 
  is $\opinf{\isone}$-closed. 
\item
  \textbf{Coinduction principle}: to prove that all judgments in a set $\sjone$ belong to the coinductive interpretation, 
  show that $\sjone$ is $\opinf{\isone}$-consistent. Indeed, if $\sjone$ is $\opinf{\isone}$-consistent, then
  \begin{align*}
    \sjone\subseteq\opinf{\isone}(\sjone)&\Longrightarrow \sjone\in\{\sjtwo\;|\;\sjtwo\subseteq\opinf{\isone}(\sjtwo)\}\\
       &\Longrightarrow\sjone\subseteq\bigcup\{\sjtwo\;|\;\sjtwo\subseteq\opinf{\isone}(\sjtwo)\}=\gfp{\opinf{\isone}}.
  \end{align*}
\end{varitemize}

\begin{example}[Finite and Infinite Trees]
  Let us consider all those (finite and infinite) trees whose nodes are labelled with
  a symbol from the alphabet $\{\circ,\bullet\}$, Example of finite trees are
  $$
  \begin{array}{c}
  \circ\langle\bullet,\circ\rangle\\
  \bullet\langle\circ\langle\circ,\bullet\rangle,\bullet\langle\circ,\bullet\rangle,\circ\rangle\\
  \bullet
  \end{array}
  $$
  while an example of an infinite tree is the (unique!) tree $t$ such that $t=\circ\langle t,t,t\rangle$.
  Let $\unone$ the set of judgments in the form $t\Downarrow$, where $t$ is a possibly infinite tree
  as above. Moreover, let $\isone$ be the inference system composed by all the instances of the
  two rules below:
  \begin{align*}
    \urule{}{\bullet\Downarrow}{} &\quad& 
    \urule{}{\circ\Downarrow}{} &\quad& 
    \urule{t_1\Downarrow,\ldots, t_n\Downarrow}{\bullet\langle t_1,\ldots, t_n\rangle\Downarrow}{} &\quad&
    \urule{t_1\Downarrow,\ldots, t_n\Downarrow}{\circ\langle t_1,\ldots, t_n\rangle\Downarrow}{}
  \end{align*}
  The inductive interpretation of $\isone$ contains all judgements $t\Downarrow$ where $t$ is \emph{finite} tree,
  while its coinductive interpretation contains also all judgments $t\Downarrow$ where $t$ is \emph{infinite}.
\end{example}
\section{Syntax and Preliminary Definitions}\label{sec:syntax}
In this section we introduce  the syntax of $\LOP$, a language of lambda terms with binary choice introduced
by de'Liguoro and Piperno~\cite{deLigPip95}. This is the language whose probabilistic semantics is the main
topic of this paper.

The most important syntactic category is certainly those of terms. Actually, $\LOP$ is nothing more
than the usual untyped and pure lambda calculus, endowed with a binary choice operator $\oplus$ which
is meant to represent nondeterministic choice.

\begin{definition}[Terms]
Let $\setvar=\{\varone,\vartwo,\ldots\}$ be a denumerable set of variables.
The set $\LOP$ of \emph{term expressions}, or \emph{terms} is the smallest set such that:
\begin{varenumerate}
\item if $\varone\in\setvar$ then $\varone\in\LOP$;
\item if $\varone\in\setvar$ and $\termone\in\LOP$, then $\abstr{\varone}{\termone}\in\LOP$;
\item if $\termone,\termtwo\in\LOP$ then $\app{\termone}{\termtwo}\in\LOP$;
\item if $\termone,\termtwo\in\LOP$ then $\ps{\termone}{\termtwo}\in\LOP$.
\end{varenumerate}
Terms are ranged over by metavariables like $\termone,\termtwo,\termthree$.
\end{definition}
Terms, as usual, are considered modulo renaming of bound variables.
The set of free variables of a term $\termone$ is indicated as $\FV{\termone}$ and is
defined as usual. A term $\termone$ is \emph{closed} if $\FV{\termone}=\emptyset$. 
The (capture-avoiding) substitution of $\termtwo$ for the free occurrences of $\varone$ in $\termone$ is denoted
$\subst{\termone}{\varone}{\termtwo}$. Unless otherwise stated, all results 
in this paper hold only for \emph{programs}, that is to say for closed terms. Values are defined in a standard way:
\begin{definition}[Values]
A term is a \emph{value} if it is a variable or a lambda abstraction. We will call $\val$ 
the set of all values. Values are ranged over by metavariables like $\valone,\valtwo,\valthree$.
\end{definition}
The reduction relation $\rightarrow$ considered in~\cite{deLigPip95} is obtained by extending usual $\beta$-reduction with
two new reduction rules, namely $\ps{\termone}{\termtwo}\rightarrow\termone$ and $\ps{\termone}{\termtwo}\rightarrow\termtwo$, which
can be applied in every possible context. In this paper, following
Plotkin~\cite{Plotkin75}, we concentrate on weak reduction: computation can only take
place in applicative (or choice) contexts.

\begin{notation}
In the following, we sometimes need to work with finite sequences of terms. A sequence $\termone_1,\ldots,\termone_n$ 
is denoted as $\mul{\termone}$. This notation can be used to denote sequences obtained from other 
sequences and terms, e.g., $\ps{\termone}{\mul{\termtwo}}$
is $\ps{\termone}{\termtwo_1},\ldots,\ps{\termone}{\termtwo_n}$ whenever $\mul{\termtwo}$ is
$\termtwo_1,\ldots,\termtwo_n$.
\end{notation}
\subsection{Distributions}
In the probabilistic semantics we will endow $\LOP$ with, a program reduces not to a single value 
but rather to a distribution of possible observables, i.e. to a function assigning a probability
to any value. This way, all possible outputs of all binary choices are taken into account, each 
with its own probability. Divergence is indeed a possibility in the untyped setting and
our definition reflects it.
\begin{definition}[Distributions]
\begin{varenumerate}
\item 
  A \emph{probability distribution} is a function 
  $\distone: \val\rightarrow\RR_{[0,1]}$ such that 
  $\sum_{\valone\in\val} \distone(\valone)\leq 1$.
  $\pdists$ denotes the set of all probability distributions.
\item 
  A \emph{proper probability distribution} is a probability distribution such
  that $\sum_{\valone\in\val} \distone(\valone)=1$.
\item 
  Given a probability distribution $\distone$, its \emph{support} $\supp{\distone}$ is the
  subset of $\val$ whose elements are values to which $\distone$ attributes positive
  probability.
\item
  Given a probability distribution $\distone$, its \emph{sum} $\sumd{\distone}$ is simply
  $\sum_{\valone\in\val} \distone(\valone)$.
\end{varenumerate}
\end{definition}
The notion of a probability distribution as we give it is general enough to capture
the semantics of those terms which have some nonnull probability of diverging.
We use the expression $\{\valone_{1}^{\probone_1},\ldots,\valone_{k}^{\probone_k}\}$ 
to denote the probability distribution with finite support $\distone$ 
defined as
$\distone(\valone)=\sum_{\valone_i=\valone}\probone_i$.
Please observe that $\sumd{\distone}=\sum_{i=1}^k\probone_i$.

Sometimes we need to compare distinct distributions. The natural way to do that is just by lifting the
canonical order on $\RR$ up to distributions, pointwise:
\begin{definition}\label{def:orddist}
$\distone\leq \disttwo$ iff $\distone(\valone)\leq \disttwo(\valone)$ for every value $\valone$.
\end{definition}
The structure $(\pdists,\leq)$ is a partial order, but not a lattice~\cite{DaveyPriestley}: the join of two distributions
$\distone$ and $\disttwo$ does not necessarily exist in $\pdists$. However, $(\pdists,\leq)$ is a complete meet-semilattice, since meets are always
guaranteed to exist: for every $\sjone\subseteq\pdists$, the greatest lower bound of distributions
in $\sjone$ is itself a distribution. Please observe, on the other hand, that all functions
from $\val$ to $\RR_\infty$ actually form a complete lattice. And that $\sumd{(\cdot)}$ as a function
from those functions to $\RR_{\infty}$ is actually a complete lattice homomorphism.

\section{Call-by-Value}\label{sec:cbv}
In this section, four ways to give a probabilistic semantics to $\LOP$ are introduced, all of them
following the so-called call-by-value passing discipline. 

A (weak) call-by-value notion of reduction can be obtained from ordinary
reduction by restricting it in such a way that only values are passed to functions. Accordingly, choices
are only made when both alternatives are themselves values:
\begin{definition}[Call-by-value Reduction]\label{def:leftRedV}
  Leftmost reduction $\rv$ is the least binary relation between $\LOP$ and $\LOP^*$
  such that:
  \begin{align*}
    (\abstr{\varone}{\termone}){\valone}&\rv\subst{\termone}{\varone}{\valone} &
    \ps{\valone}{\valtwo}&\rv\valone,\valtwo\\
    \termone\termtwo&\rv\mul{\termthree}\termtwo\qquad\mbox{if $\termone\rv\mul{\termthree}$} &
    \ps{\termone}{\termtwo}&\rv\ps{\mul{\termthree}}{\termtwo}\qquad\mbox{if $\termone\rv\mul{\termthree}$}\\
    \valone\termone&\rv\valone\mul{\termtwo}\qquad\mbox{if $\termone\rv\mul{\termtwo}$} &
    \ps{\valone}{\termone}&\rv\ps{\valone}{\mul{\termtwo}}\qquad\mbox{if $\termone\rv\mul{\termtwo}$}
  \end{align*}
  where $\valone,\valtwo\in\val$.
\end{definition}
Please, notice that reduction \emph{is not probabilistic}. In fact, reduction is a relation between
terms and unlabeled sequences of terms without any reference to probability. Informally,
if $\termone\rv\termtwo_1\ldots\termtwo_n$ means that $\termone$ rewrites in one step to every $\termtwo_i$
with the \emph{same} probability $1/n$. Clearly, $n\in\{1,2\}$ whenever 
$\termone\rv\termtwo_1\ldots\termtwo_n$. Notice again how the evaluation of \emph{both} branches
of a binary choice is done \emph{before} performing the choice itself. One the one hand, this
is very much in the style of call-by-value evaluation. On the other, a more standard notion of choice,
which is performed before evaluating the brances can be easily encoded as follows:
$$
\termone+\termtwo=(\ps{(\abstr{\varone}{\abstr{\vartwo}{\varone}})}{(\abstr{\varone}{\abstr{\vartwo}{\vartwo}})})
   (\abstr{\varthree}{\termone})(\abstr{\varthree}{\termtwo})(\abstr{\varfour}{\varfour}),
$$
where $\varthree$ does not appear free in $\termone$ nor in $\termtwo$.

\subsection{CbV Small-step Semantics} \label{sec:cbvsmall}
Following the general methodology described in ~\cite{Leroy}, we model separately convergence and divergence. 
Finite computations are, as usual, inductively defined, while divergence can be captured by
interpreting a different set of rules coinductively. Both definitions give some \emph{quantitative} information 
about the dynamics of any term $\termone$: either a distribution of possible outcomes is associated to $\termone$,
or the probability of divergence of $\termone$ is derived. Both induction and coinduction can be used to
characterize the distribution of values a term evaluates to. In an inductive characterization, one
is allowed to underapproximate the target distribution, but then takes an upper bound of all the
approximations. In a coinductive characterization, on the other hand, one can naturally overapproximates
the distribution, then taking a lower bound.

\subsubsection{Inductive CbV Small-Step Semantics for Convergence}\label{sect:isssfc}
A small-step semantics for convergence is captured here by way of a binary relation $\ssemv{}{}$ between 
terms in $\LOP$ and distributions. The relation $\ssemv{}{}$ is defined as the inductive interpretation
of the inference system whose rules are all possible instances of the following ones: 
$$
\zrule
    {
      \ssemv{\termone}{\emdist}
    }
    {\sev}
\qquad
\zrule
    {
      \ssemv{\valone}{\{\valone^1\}}
    }
    {\svv}
\qquad
\brule
    {
      \termone\rv\mul{\termtwo}
    }
    {
      \ssemv{\termtwo_i}{\distone_i}
    }
    {
      \ssemv{\termone}{\sum_{i=1}^n\frac{1}{n}\distone_i}
    }
    {\smv}
$$
As usual, $\mul{\termtwo}$ stands for the sequence $\termtwo_1,\ldots\termtwo_n$. Since the
relation $\ssemv{}{}$ is inductively defined, any proof of judgments involving $\ssemv{}{}$
is a \emph{finite} object. If $\derone$ is such a proof for $\ssemv{\termone}{\distone}$,
we write $\derone:\ssemv{\termone}{\distone}$. The proof $\derone$ is said to be
\emph{structurally smaller or equal to} another proof $\dertwo$ if the number
of rule instances in $\derone$ is smaller or equal to the number of rule instances
in $\dertwo$. In this case, we write $\derone\leqder\dertwo$.

First of all, observe that $\ssemv{}{}$ is not a function: many different distributions
can be put in correspondence with the same term $\termone$. Moreover, there is
one distribution $\distone$ such that $\ssemv{\termone}{\distone}$ always holds, independently on $\termone$,
namely $\emdist$. Actually, rule $\sev$ allows you to ``give up'' while looking for a distribution
for $\termone$ and conclude that $\termone$ is in relation with $\emdist$. In other words,
$\ssemv{}{}$ is not meant to be a way to attribute \emph{one} distribution to every term, but rather
to find all finitary approximants of the (unique!) distribution we are looking for (see Definition~\ref{def:ssv}).

The set of the distributions that a  term evaluates to  is a direct set:
\begin{lemma}\label{lemma:linord}
For every term $\termone$, if $\ssemv{\termone}{\distone}$ and $\ssemv{\termone}{\disttwo}$,
then there exists a distribution $\distthree$ such that $\ssemv{\termone}{\distthree}$ with 
$\distone\geq\distthree$ and $\disttwo\geq\distthree$.
\end{lemma}
\begin{proof}
  By induction on the structure of derivations for
  $\ssemv{\termone}{\distone}$.
  \begin{varitemize}
  \item If $\ssemv{\termone}{\emptyset}$ then $\distthree=\disttwo$;
  \item If $\ssemv{\termone}{\{\valone^{1}\}}$ we have that $\disttwo=\emptyset$, or $\disttwo=\{\valone^{1}\}$, 
     then $\distthree=\distone$;
  \item If $\ssemv{\termone}{\termtwo_1,\ldots,\termtwo_k}$ and $\ssemv{\termtwo_i}{\distone_i}$ for $i=1,\ldots,k$,  
     there are some cases:
  \begin{varitemize}
  \item If $\disttwo=\emptyset$, then $\distthree=\distone$;
  \item If $\disttwo >\emptyset$, then $\disttwo=\sum_{i=1}^{k}\frac{1}{k}\disttwo_i$, where  $\ssemv{\termtwo_i}{\disttwo_i}$ 
    for $i=1,\ldots,k$. Now, by inductive hypothesis, there exist distributions $\distthree_i$  such that 
    $\ssemv{\termtwo_i}{\distthree_i}$ and $\distone_i,\disttwo_i\leq \distthree_i$ for $i=1,\ldots,k$. We have 
    $\distthree=\sum_{i=1}^{k}\frac{1}{k}\distthree_i$, and by definition 
    $\distone\leq\distthree$ and $\disttwo\leq\distthree$.
  \end{varitemize}
  \end{varitemize}
This concludes the proof.
\end{proof}

We are now ready to define what \emph{the} small-step semantics of any term is:

\begin{definition}\label{def:ssv}
The (call-by-value) small-step semantics of a lambda term $\termone\in\LOP$  is the distribution
$\Ssemv{\termone}$ defined as $\sup_{\ssemv{\termone}{\distone}}\distone$. 
\end{definition}

Please observe that such a distribution
is always guaranteed to exist, precisely because of Lemma~\ref{lemma:linord}. Indeed
$$
\sumd{\left(\sup_{\ssemv{\termone}{\distone}}\distone\right)}=\sup_{\ssemv{\termone}{\distone}}\left(\sumd{\distone}\right)\leq 1,
$$
since $\sumd{\distone}$ is at most $1$, by hypothesis.

\begin{example}
The term $\app{(\abstr{\varone}{\varone})}{(\abstr{\varone}{\varone})}$ evaluates to $\distthree=\emptyset$, by means of rules
$\sev$, and to $\disttwo=\{{\abstr{\varone}{\varone}}^{1}\}$, by means of rules $\smv$ and $\svv$. By definition, $\Ssemv{\termone}=\disttwo$.
\end{example}

\subsubsection{Coinductive CbV Small-step Semantics for Divergence}
Divergence is captured by another, \emph{coinductively} defined binary relation $\sdsemv{}{}$ between $\LOP$ and 
$\RRp{0}{1}$. Rules defining the underlying inference systems are those obtained from
$$
\czrule
    {
      \sdsemv{\valone}{0}
    }
    {\dvv}
\qquad
\cbrule
    {
      \termone\rv\mul{\termtwo}
    }
    {
	 \sdsemv{\termtwo_i}{\probone_i}
    }
    {
      \sdsemv{\termone}{\sum_{i=1}^n\frac{1}{n}\probone_i}
    }
    {\dmv}
$$
where $\mul{\termtwo}$ is $\termtwo_1,\ldots\termtwo_n$.

The relation $\sdsemv{}{}$ deals naturally with infinite computations, being coinductively defined. This allows
to derive the divergence probability $\probone$ of a term in $\LOP$. Rules $\dvv$ and $\dmv$ can be read
as follows: values diverge with probability $0$, while probability of divergence for a term $\termone$ is equal 
to the normalized sum of its reducts' probabilities of divergence. The following example stresses a crucial point:

\begin{example}\label{exa:cbvssdiv}
Let us consider the well-know diverging term $\Omega=\Delta\Delta$, where $\Delta=\abstr{\varone}{\varone\varone}$.
We would like to be able to prove that $\sdsemv{\Omega}{1}$, namely that $\Omega$ diverges with probability $1$. 
Doing that formally requires proving that some sets of judgments $\sjone$ including $\sdsemv{\Omega}{1}$ is consistent
with respect to $\sdsemv{}{}$. Actually we can choose $\sjone$ as $\{\sdsemv{\Omega}{1}\}$, since
$\sdsemv{\Omega}{1}$ can be easily derived from itself by rule $\dmv$. The trouble is that in the same way we 
can derive $\sdsemv{\Omega}{\probone}$ for every possible $0\leq\probone\leq 1$! So, in a sense,
$\sdsemv{}{}$ is inconsistent.
\end{example}

Example~\ref{exa:cbvssdiv} shows that not all probabilities attributed to terms via $\sdsemv{}{}$ are accurate. Actually,
a good definition consists in taking the \emph{divergence probability} of any term $\termone$, $\Sdsemv{\termone}$ simply
as $\sup_{\sdsemv{\termone}{\probone}}\probone$.
As an example, $\Sdsemv{\Omega}=1$, since $\sdsemv{\termone}{1}$ and, clearly, one cannot go beyond $1$.

\subsubsection{Coinductive CbV Small-Step Semantics for Convergence}\label{sec:cobsv}
If one takes the inductive semantics from Section~\ref{sect:isssfc}, drops
$\sev$ and interpret everything coinductively, what came out is an alternative 
semantics for convergence:
$$
\czrule
    {
      \cssemv{\valone}{\{\valone^1\}}
    }
    {\csvv}
\qquad
\cbrule
    {
      \termone\rv\mul{\termtwo}
    }
    {
      \cssemv{\termtwo_i}{\distone_i}
    }
    {
      \cssemv{\termone}{\sum_{i=1}^n\frac{1}{n}\distone_i}
    }
    {\csmv}
$$
Interpreting everything coinductively has the effect of allowing infinite computations to be
modeled. But this allows to ``promise'' to reach a certain distribution without really being
able to fullfill this:
\begin{example}
Consider again $\Omega$. One would like to be able to prove that $\cssemv{\Omega}{\emdist}$.
Unfortunately, the coinductive interpretation of
the formal system above contains $\cssemv{\Omega}{\distone}$ \emph{for every} 
distribution $\distone$, as can be easily verified
\end{example}
The solution to the problem highlighted by the example above is just defining
the coinductive semantics of any term $\termone$ just as $\CSsemv{\termone}=\inf_{\cssemv{\termone}{\distone}}\distone$.
Clearly, $\CSsemv{\Omega}=\emdist$.

\subsection{CbV Big-Step Semantics}\label{sec:cbvbig}

An alternative style to give semantics to programming languages is the so-called big-step semantics.
Big-step semantics is more compositional than small-step: the meaning of a term can be obtained
from the meanings of its sub-terms.

Probabilistic big step semantics for $\LOP$ can be given by a binary relation 
$\bsemv{}{}$ between $\LOP$ and distributions. It is the coinductive interpretation
of (all instances of) the following rules:
$$
\czrule
    {
      \bsemv{\valone}{\{\valone^1\}}
    }
    {\bvv}
\qquad
\ctrule
    { 
      \bsemv{\termone}{\distone}
    }
    {
      \bsemv{\termtwo}{\disttwo}
    }
    {
      \{\bsemv{\subst{\termfour}{\varone}{\valone}}{\distthree_{\termfour,\valone}}\}_{\abstr{\varone}{\termfour}\in\supp{\distone},\valone\in\supp{\disttwo}}
    }
    {
      \bsemv{\app{\termone}{\termtwo}}{\sum_{\abstr{\varone}{\termfour}\in\supp{\distone},\valone\in\supp{\disttwo}}
       \distone(\abstr{\varone}{\termfour})\cdot \disttwo(\valone)\cdot \distthree_{\termfour,\valone}}
    }
    {\bav
    }
$$
$$
\cbrule
    {
      \bsemv{\termone}{\distone}
    }
    {
      \bsemv{\termtwo}{\disttwo}
    }
    {
      \bsemv{\ps{\termone}{\termtwo}}{\frac{1}{2}\cdot \distone\cdot \sumd{\disttwo}+\frac{1}{2}\cdot \disttwo\cdot\sumd{\distone}}
    }
    {\bsv
    }
$$
The most interesting rule is definitely $\bav$: to give semantics to an application
$\app{\termone}{\termtwo}$, we first give semantics to $\termone$ and $\termtwo$,
obtaining two distributions $\distone$ and $\disttwo$, respectively. Then, for every
$\abstr{\varone}{\termfour}$ in $\supp{\distone}$ and for every $\valone$  in $\supp{\disttwo}$,
we evaluate $\subst{\termfour}{\varone}{\valone}$, obtaining some distributions
$\distthree_{\termfour,\valone}$. The meaning of $\app{\termone}{\termtwo}$ is nothing more that
the sum of all such distributions $\distthree_{\termfour,\valone}$, each weighted by the probability
of getting $\abstr{\varone}{\termfour}$ and $\valone$ in $\distone$ and $\disttwo$, respectively.
This way of defining the big-step semantics of applications can be made simpler, at the price of
making the task of deriving semantic assertions harder, by replacing the premise 
$
\{\bsemv{\subst{\termfour}{\varone}{\valone}}{\distthree_{\termfour,\valone}}\}_{\abstr{\varone}{\termfour}\in\supp{\distone},\valone\in\supp{\disttwo}}
$
of rule $\bav$ with
$
\{\bsemv{\subst{\termfour}{\varone}{\valone}}{\distthree_{\termfour,\valone}}\}_{\termfour\in\LOP,\valone\in\val}.
$

Another interesting rule is $\bsv$. Please observe how the distributions
$\distone$ and $\disttwo$ must be normalized by $\sumd{\disttwo}$ and $\sumd{\distone}$ (respectively) when
computing the result. This reflects call-by-value evaluation: if any of $\termone$ and $\termtwo$ diverges, their
sum must diverge, too.

Like $\ssemv{}{}$ and $\cssemv{}{}$, the relation $\bsemv{}{}$ is not a function: many possible distributions can be
assigned to the same term $\termone$, in particular when $\termone$ (possibly) diverges. In particular,
distributions which somehow overapproximate the ``real one'' can always be attributed to $\termone$ by
the rules above. Thus, we define \emph{the} big-step semantics of each lambda-term in the following way:

\begin{definition}\label{def:bsv}
The (call-by-value) coinductive big-step semantics of a lambda term $\termone\in\LOP$  is the distribution
$\Bsemv{\termone}$ defined as $\inf_{\bsemv{\termone}{\distone}}\distone$. 
\end{definition}

\begin{example}\label{exa:omega}
Let us consider again $\Omega$, from Example~\ref{exa:cbvssdiv}.
As it can be verified, $\bsemv{\Omega}{\distone}$ for all possible distributions  
$\distone$. To formally prove that, we need to find a consistent set $\sjone$ (with respect
to $\bsemv{}{}$) containing all judgments in the form
$\bsemv{\Omega}{\distone}$. $\sjone$ is actually the set
$$
\{\bsemv{\Delta}{\{\Delta^1\}}\}\cup\{\bsemv{\Omega}{\distone}\mid\distone:\val\rightarrow\RRp{0}{1}\}.
$$
Clearly, $\sjone$ is consistent, since any judgment in $\sjone$ can be obtained from
other judgments in $\sjone$ in one deduction step: 
\begin{varitemize}
\item
  $\bsemv{\Delta}{\{\Delta^1\}}$ by rule $\bvv$.
\item
  From $\bsemv{\Delta}{\{\Delta^1\}}$ and $\bsemv{\Omega}{\distone}$, one
  can easily derive $\bsemv{\Omega}{\distone}$ by rule $\bav$.
\end{varitemize}
As a consequence, $\Bsemv{\Omega}=\inf_{\bsemv{\Omega}{\distone}}\distone=\emptyset$.
\end{example}

One may wonder whether an \emph{inductive} big step semantics can be defined for $\LOP$. The answer
is positive: one only needs to add a rule attributing the empty distribution to any terms, in the
spirit of the small-step inductive semantics from Section~\ref{sect:isssfc}. In other words, we
obtain the system
$$
\zrule
    {
      \bsemv{\termone}{\{\emdist\}}
    }
    {\ibev}
\qquad
\trule
    { 
      \bsemv{\termone}{\distone}
    }
    {
      \ibsemv{\termtwo}{\disttwo}
    }
    {
      \{\ibsemv{\subst{\termfour}{\varone}{\valone}}{\distthree_{\termfour,\valone}}\}_{\abstr{\varone}{\termfour}\in\supp{\distone},\valone\in\supp{\disttwo}}
    }
    {
      \ibsemv{\app{\termone}{\termtwo}}{\sum_{\abstr{\varone}{\termfour}\in\supp{\distone},\valone\in\supp{\disttwo}}
       \distone(\abstr{\varone}{\termfour})\cdot \disttwo(\valone)\cdot \distthree_{\termfour,\valone}}
    }
    {\ibav
    }
$$
$$
\zrule
    {
      \ibsemv{\valone}{\{\valone^1\}}
    }
    {\ibvv}
\qquad
\brule
    {
      \ibsemv{\termone}{\distone}
    }
    {
      \ibsemv{\termtwo}{\disttwo}
    }
    {
      \ibsemv{\ps{\termone}{\termtwo}}{\frac{1}{2}\cdot \distone\cdot \sumd{\disttwo}+\frac{1}{2}\cdot \disttwo\cdot\sumd{\distone}}
    }
    {\ibsv
    }
$$
As can be expected, the inductive big step semantics $\IBsemv{\termone}$ of
an term $\termone$ is simply $\sup_{\ibsemv{\termone}{\distone}}$. This is similar to
the semantics considered by Jones~\cite{JonesPhD}.

\subsection{Divergence and Convergence in CbV Small-step Semantics}\label{sec:cbvDvsC}
In the last two sections, various operational semantics for both convergence and
divergence have been introduced. Clearly, one would like them to be essentially
equivalent, i.e., one would like them to attribute the same meaning to programs.

In this section, divergence and convergence small-step semantics will be compared and
proved equivalent: the probability of divergence of $\termone$ obtained through
$\sdsemv{}{}$ will be proved to be equal to $1-\sumd{\Ssemv{\termone}}$. 
In Section~\ref{sec:cbvSvsB} small-step and big-step semantics
for convergence will be proved to produce identical outcomes.

The first step consists in proving that $1-\sumd{\Ssemv{\termone}}$ is a lower
bound to divergence probability of any term $\termone$.
\begin{theorem}\label{th:cbvconvdiv}
For every term $\termone$, $\sdsemv{\termone}{1-\sumd{\Ssemv{\termone}}}$.
\end{theorem}
\begin{proof}
We can prove that all judgments $\sdsemv{\termone}{1-\sumd{\Ssemv{\termone}}}$
belong to the coinductive interpretation of the underlying formal system $\isone$.
To do that, we need to prove that the set $\sjone$ of all these judgements
is consistent, i.e. that $\sjone\subseteq\opinf{\isone}(\sjone)$.
This amounts to show that if $\juone\in\sjone$, then there is a derivation
for $\juone$ whose immediate premises are themselves in $\sjone$. Let's
distinguish two cases:
\begin{varitemize}
  \item
    If $\termone$ is a value $\valone$, then $\sdsemv{\termone}{0}$ and
    $\Ssemv{\termone}=\{\valone^1\}$. The thesis easily follows.
  \item
    If $\termone$ is not a value, then $\termone\rv\mul{\termtwo}$, with 
    $\mul{\termtwo}=\termtwo_1,\ldots,\termtwo_n$. Now, consider the judgements 
    $\sdsemv{\termtwo_i}{1-\sumd{\Ssemv{\termtwo_i}}}$, with $i\in\interv{1}{n}$: 
    they are all in the set $\sjone$. Finally, consider the judgement 
    $\sdsemv{\termone}{\sum_{i=1}^{n}\frac{1}{n}{(1-\sumd{\termtwo_i})}}$: 
    it is in $\opinf{\isone}{(\sjone)}$ because of the presence of rule $\dmv$:
    $$
    \cbrule
        {
          \termone\rv\mul{\termtwo}
        }
        {
          \sdsemv{\termtwo_i}{\probone_i}
        }
        {
          \sdsemv{\termone}{\sum_{i=1}^n\frac{1}{n}\probone_i}
        }
        {\dmv}
    $$    
    It remains to show that $\sumd{\Ssemv{\termone}}=\sum_{i=1}^{n}\frac{1}{n}\sumd{\Ssemv{\termtwo_i}}$:
    \begin{align*}
      \sumd{\Ssemv{\termone}}&=\sumd{\sup_{\ssemv{\termone}{\distone}}\distone}
      =\sumd{\sup_{\ssemv{\termtwo_i}{\disttwo_i}}\left(\sum_{i=1}^n\frac{1}{n}\disttwo_i\right)}
      =\sumd{\left(\sum_{i=1}^n\frac{1}{n}\sup_{\ssemv{\termtwo_i}{\disttwo_i}}\disttwo_i\right)}\\
      &=\sumd{\left(\sum_{i=1}^n\frac{1}{n}\Ssemv{\termtwo_i}\right)}
      =\sum_{i=1}^n\frac{1}{n}\sumd{\Ssemv{\termtwo_i}}
    \end{align*}
\end{varitemize}
This concludes the proof.
\end{proof}
We need something else, however, namely that summing some divergence probability and
the convergence probability obtained through convergence semantics, we cannot go beyond $1$:
\begin{proposition}\label{prop:limitcbv}
If $\sdsemv{\termone}{\probone}$ and $\ssemv{\termone}{\distone}$, then 
$\probone+\sumd{\distone}\leq 1$.
\end{proposition}
\begin{proof}
The proof goes by induction on the structure of the derivation
for $\ssemv{\termone}{\distone}$:
\begin{varitemize}
\item
  If the only rule in the derivation is
  $$
  \zrule
  {
    \ssemv{\termone}{\emdist}
  }
  {}
  $$
  then $\distone=\emdist$ and $\sumd{\distone}=0$. As a consequence,
  $\probone+\sumd{\distone}=\probone\leq 1$ by definition of the divergence relation.
\item
  If the only rule in the derivation is
  $$
  \zrule
  {
    \ssemv{\valone}{\{\valone^1\}}
  }
  {}
  $$
  then $\sumd{\distone}=1$ and 
  $\sdsemv{\termone}{0}$. As a consequence
  $\probone+\sumd{\distone}\leq 1$.
\item
  If the derivation has the form
  $$
  \brule
  {
    \termone\rv\mul{\termtwo}
  }
  {
    \ssemv{\termtwo_i}{\distone_i}
  }
  {
    \ssemv{\termone}{\sum_{i=1}^n\frac{1}{n}\distone_i}
  }
  {}
  $$  
  then we apply the induction hypothesis on each  $\ssemv{\termtwo_i}{\distone_i}$ for $i\in\interv{1}{n}$.
 For each $i\in\interv{1}{n}$, $\sum \distone_i\leq 1-\probone_i$, where $\probone_i$ is such that $\sdsemv{\termtwo_i}{\probone_1}$.
 Then we have:
 \begin{align*}
 \sum\distone&=\sum\sum_{i=1}^{n}\frac{1}{n}\distone_i=\sum_{i=1}^{n}\frac{1}{n}\sum\distone_i\stackrel{i.h.}{\leq}\sum_{i=1}^{n}\frac{1}{n}(1-\probone_i)\\
 &=1-\sum_{i=1}^{n}\frac{1}{n}\probone_i=1-\probone.
 \end{align*}
 \end{varitemize}
This concludes the proof.
\end{proof}
Everything can be now glued together as follows, exploiting the density of real numbers:
\begin{corollary}
For every $\termone$, $\Sdsemv{\termone}+\sumd{\Ssemv{\termone}}=1$.
\end{corollary}
\begin{proof}
$\Sdsemv{\termone}+\sumd{\Ssemv{\termone}}\geq 1$ by Theorem \ref{th:cbvconvdiv}.
Suppose, by way of contradiction, that $\Sdsemv{\termone}+\sumd{\Ssemv{\termone}}\gt 1$.
This implies $\sdsemv{\termone}{\probone}$ where $\probone+\sumd{\Ssemv{\termone}}\gt 1$. 
This, in turn, implies that $\probone+\sumd{\distone}\gt 1$ for some $\ssemv{\termone}{\distone}$. And this is not
possible by Proposition~\ref{prop:limitcbv}.
\end{proof}

\subsection{Relating the Various Definitions for Convergence}\label{sec:cbvSvsB}
Our goal in this section consists in proving that the four
ways we have defined the semantics of a term in $\LOP$ 
(inductive and coinductive, big-step and small-step semantics) attribute identical
meanings to any term. One possibility could be to proceed by ``proving''
three edges of the following diagram:
$$
\xymatrix{
\Ssemvu \ar@{-}[r] \ar@{-}[d] & \CSsemvu\\
\IBsemvu \ar@{-}[r] & \Bsemvu \ar@{-}[u]
}
$$
The vertical edges relate two formulations given with identical ``inductive flavors''
but differing as to whether they are big-step or small-step. Conversely, horizontal
edges put in correspondence two formulations which are both big-step or small step, but which differ
as to which kind of interpretation is taken over the same set of rules. Horizontal
edges are definitely more interesting to be proved, but vertical ones are important,
too. In order to avoid tedious and long (but not necessarily informative) proofs,
we only prove the diagonal edge shown below in this paper:
$$
\xymatrix{
\Ssemvu \ar@{~}[rd] & \CSsemvu\\
\IBsemvu & \Bsemvu 
}
$$
This proof shows the difficulties of both ``vertical'' and ``horizontal'' edges. 

Before embarking in the proof of this result, a brief explanation
of the architecture of the proof is maybe useful. 

Consider any term $\termone$, and define two sets of distributions
$\Ssemvs{\termone}$ and $\Bsemvs{\termone}$ as, respectively, the
sets of probability distributions which can be attributed to $\termone$
in small-step semantics and big-step semantics, respectively:
\begin{align*}
  \Ssemvs{\termone}&=\{\distone\mid\ssemv{\termone}{\distone}\};\\
  \Bsemvs{\termone}&=\{\distone\mid\bsemv{\termone}{\distone}\}.
\end{align*}
We proceed in two steps:
\begin{varitemize}
\item
  First of all, we prove that big-step semantics \emph{dominates} small-step
  semantics, namely that $\distone\geq\disttwo$ whenever $\distone\in\Bsemvs{\termone}$
  and $\disttwo\in\Ssemvs{\termone}$. This way, we are sure that
  $\Bsemv{\termone}\geq\Ssemv{\termone}$. Details are in Section~\ref{sect:domination}
  below.
\item
  Then, we prove that small step semantics can itself be derived using
  big-step semantics, namely that $\Ssemv{\termone}\in\Bsemvs{\termone}$. 
  This way, we immediately obtain that $\Bsemv{\termone}\leq\Ssemv{\termone}$,
  since $\Bsemv{\termone}$ is a minorizer of distributions in $\Bsemvs{\termone}$.
  Details are in Section~\ref{sect:ssinbs}.
\end{varitemize}
As a consequence, $\Ssemvs{\termone}=\Bsemvs{\termone}$. Figure~\ref{fig:ssvsbs}
illustrates the architecture of the proof.

\begin{figure*}
\begin{center}
\includegraphics[scale=1.4]{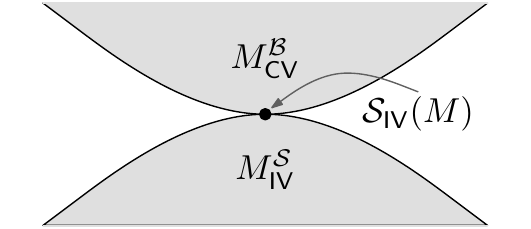}
\end{center}
\caption{The Overall Picture}\label{fig:ssvsbs}
\end{figure*}
\subsubsection{Big-step Dominates Small-step}\label{sect:domination}
The fact any distribution obtained through big-step semantics is bigger than any distribution
obtained through small-step semantics can be proved by induction on the structure of (finite!)
derivations for the latter. Compulsory to that, however, is proving that whenever 
$\ssemv{\termone}{\distone}$ and $\termone$ is, say, a sum $\ps{\termtwo}{\termthree}$,
then appropriate judgments $\ssemv{\termtwo}{\disttwo}$ and $\ssemv{\termthree}{\distthree}$
can be derived. Similarly for applications. 

The following lemma formalizes these ideas and it is a technical tool for Proposition~\ref{prop:leqbs}.

\begin{lemma}\label{lemma:struct}
If $\ssemv{\termone}{\distone}$, then at least
one of the following conditions hold:
\begin{varenumerate}
\item\label{caseoneL}
  $\distone=\emptyset$;
\item\label{casetwoL}
  $\termone$ is a value $\valone$ and $\distone=\{\valone^1\}$;
\item\label{casethreeL}
  $\termone$ is an application $\app{\termtwo}{\termthree}$ and
  there are (finite) distributions $\disttwo$ and
  $\distthree$ and 
  for every $\abstr{\varone}{\termfour}\in\supp{\disttwo} $ and $\valone\in\supp{\distthree}$ a distribution
  $\distfour_{\termfour,\valone}$ 
  such that:
  \begin{varenumerate}
  \item\label{case3oneL}
    $\dertwo:\ssemv{\termtwo}{\disttwo}$, $\derthree:\ssemv{\termthree}{\distthree}$
    and $\derfour_{\termfour,\valone}:\ssemv{\subst{\termfour}{\varone}{\valone}}{\distfour_{\termfour,\valone}}$;
  \item\label{case3twoL}
    $\dertwo,\derthree,\derfour_{\termfour,\valone}\ltder\derone$;
  \item\label{case3threeL}
    $\distone\leq\sum_{\abstr{\varone}{\termfour}\in\supp{\disttwo},\valone\in\supp{\distthree}}\disttwo(\abstr{\varone}{\termfour})\cdot\distthree(\valone)\cdot\distfour_{\termfour,\valone}$.
  \end{varenumerate}
\item\label{casefourL}
  $\termone$ is a sum $\ps{\termtwo}{\termthree}$ and there are (finite) distributions
  $\disttwo$ and
  $\distthree$ such that:
  \begin{varenumerate}
  \item\label{case4oneL}
    $\dertwo:\ssemv{\termtwo}{\disttwo}$ and $\derthree:\ssemv{\termthree}{\distthree}$;
  \item\label{case4twoL}
    $\dertwo,\derthree\ltder\derone$;
  \item\label{case4threeL}
    $\distone\leq \frac{1}{2}\cdot\disttwo\cdot(\sumd{\distthree})+
    \frac{1}{2}\cdot\distthree\cdot(\sumd{\disttwo})$.   
  \end{varenumerate}
\end{varenumerate}
\end{lemma}
\begin{proof}
First of all, let us prove the following auxiliary lemma:
\begin{lemma}\label{lemma:supdist}
If $I$ is finite and $\derone_i:\ssemv{\termone}{\distone_i}$ for every $i\in I$, then
there is $\dertwo:\ssemv{\termone}{\sup_{i\in I}\distone_i}$. Moreover
$\dertwo\leqder\derone_i$ for some $i\in I$.
\end{lemma}
The proof goes by induction on the structure of the proofs
in the family $\{\derone_i\}_{i\in I}$ (which can actually be
done, since $I$ is finite).

Let's now go back to Lemma~\ref{lemma:struct}.
This is an induction on derivations for $\ssemv{\termone}{\distone}$.
The cases where $\derone$ is obtained by the rules without premises
are trivial. So, we can assume that $\derone$ has the form
$$
\brule
    {
      \termone\rv\mul{\termfive}
    }
    {
   \{   \derone_t:\ssemv{\termfive_i}{\distone_i}\}_{i\in\interv{1}{n}}
    }
    {
      \ssemv{\termone}{\sum_{i=1}^n\frac{1}{n}\distone_i}
    }
    {\smv}
$$
where $\mul{\termfive}$ is $\termfive_1,\ldots,\termfive_n$.
Let's distinguish some cases depending on
how the premise $\termone\rv\mul{\termfive}$ is derived
\begin{varitemize}
  \item
  Suppose $\termone=\app{\termtwo}{\termthree}$ and
  that $\termtwo\rv\mul{\termsix}$. Then
  $\mul{\termfive}=\app{\mul{\termsix}}{\termthree}$. From 
  the induction hypothesis applied to
  the derivations $\derone_i:\ssemv{\app{\termsix_i}{\termthree}}{\distone_i}$, we obtain that
  \begin{varitemize}
    \item
      Either all of the $\distone_i$ is $\emdist$. In this case $\distone$ is itself $\emdist$, and a 
      derivation for $\ssemv{\termone}{\emdist}$ can be defined.
    \item
      Or there is at least one among the derivation $\derone_i$ to which case \ref{casethreeL}\ applies.
      Suppose that $k_1,\ldots,k_m$ are the indices in $\interv{1}{n}$ to which case \ref{casethreeL} can
      be applied. This means that for every in $i$ in $\{1,\ldots,m\}$, there is
      a distribution $\disttwo_{i}$, 
      $\distthree_i$ 
      and for every $\abstr{\varone}{\termfour}\in{\supp{\disttwo_i}}$ and
      $\valone\in\supp{\distthree_i}$ 
      a distribution
      $\distfive^{i}_{\termfour^{i},\valone^{i}}$ such that: 
      \begin{align*}
         \dertwo_i&: \ssemv{\termsix_{k_i}}{\disttwo_{i}};\\ 
         \derthree_i&: \ssemv{\termthree}{\distthree_i};\\
         \derfive^{i}_{\termfour,\valone}&: \ssemv{\termfour\{\valone/\varone\}}{\distfive^{i}_{\termfour,\valone}}.
      \end{align*}
      Moreover, $ \dertwo_i,\derthree_i,\derfive^{i}_{g,h}\ltder\derone_{t}$,  and 
      $$
      \distone_{k_i}\leq\sum_{\abstr{\varone}{\termfour^{}}\in{\supp{\disttwo_i}}, 
        \valone^{}\in\supp{\distthree_i}}\disttwo_{i}(\termfour^{})\cdot\distthree_{i}(\valone^{})\cdot\distfive^{i}_{\termfour, \valone}.
      $$
      Now, define $\disttwo$ as the distribution
      $$
      \sum_{i=1}^m\frac{1}{n}\disttwo_i
      $$
      
      Note that $\supp{\disttwo}=\bigcup_{i=1}^{m}\supp{\disttwo_i}$. 
      Clearly, a derivation $\dertwo$ for $\ssemv{\termtwo}{\disttwo}$ can be defined such that $\dertwo\ltder\derone$: simply
      construct it from the derivations $\dertwo_i$ and some derivations for $\ssemv{\termsix_j}{\emdist}$.
      Moreover, a derivation $\derthree$ of $\ssemv{\termthree}{\distthree}$ where $\distthree=\sup_{i=1}^m\distthree_{i}$ 
      can be defined such that $\derthree\ltder\derone$: use Lemma \ref{lemma:supdist}. Observe
      that $\supp{\distthree}=\bigcup_{i=1}^{m}\supp{\distthree_i}$. Similarly,
      derivations $\derfour_{\termfour,\valone}$ can be defined for every
      $\abstr{\varone}{\termfour}\in\supp{\disttwo}$ and for every $\valone\in\supp{\distthree}$
      in such a way that $\derfour_{\termfour,\valone}\ltder\derone$,
      $\derfour_{\termfour,\valone}:\ssemv{\termfour\{\valone/\varone\}}{\distfour_{\termfour,\valone}}$
      and $\distfour_{\termfour,\valone}=\sup_{i=1}^m\distfive_{\termfour,\valone}^i$.
      Now:
      \begin{align*}
          \distone&=\sum_{i=1}^{n}\frac{1}{n}\distone_{i}\leq \sum_{i=1}^{m}\frac{1}{n}\distone_{k_i}\\
          &\leq\sum_{i=1}^{m}\frac{1}{n} \left(\sum_{\abstr{\varone}{\termfour}\in\supp{\disttwo_i},
               \valone\in\supp{\distthree_i}}\disttwo_{i}(\abstr{\varone}{\termfour})\cdot
               \distthree_{i}(\valone)\cdot\distfive^{i}_{\termfour,\valone}\right)\\
          &\leq\sum_{i=1}^{m}\frac{1}{n} \left(\sum_{\abstr{\varone}{\termfour}\in\supp{\disttwo},
               \valone\in\supp{\distthree}}\disttwo_{i}(\abstr{\varone}{\termfour})\cdot
               \distthree(\valone)\cdot\distfive^{i}_{\termfour,\valone}\right)\\
          &=\sum_{\abstr{\varone}{\termfour}\in\supp{\disttwo},
               \valone\in\supp{\distthree}}\sum_{i=1}^{m}\frac{1}{n}
               \left(\disttwo_{i}(\abstr{\varone}{\termfour})\cdot
               \distthree(\valone)\cdot\distfive^{i}_{\termfour,\valone}\right)\\
          &\leq\sum_{\abstr{\varone}{\termfour}\in\supp{\disttwo_i},\valone\in\supp{\distthree}}
               \sum_{i=1}^{m}\frac{1}{n}\left(\disttwo_{i}(\abstr{\varone}{\termfour})
               \cdot\distthree(\valone)\cdot\distfour_{\termfour,\valone}\right)\\
          &=\sum_{\abstr{\varone}{\termfour}\in\supp{\disttwo},\valone\in\supp{\distthree}}
               \disttwo(\abstr{\varone}{\termfour})\cdot\distthree(\valone)\cdot\distfour^{}_{\termfour,\valone}\\
      \end{align*}             
     
  \end{varitemize}
  \item 
    Suppose $\termone=\app{\termtwo}{\termthree}$ and that $\termthree\rv\mul{\termsix}$.  
    Then $\mul{\termfive}=\app{\termtwo}{\mul{\termsix}}$. This case is very similar to the previous one. 
    Note that since we reduce in a call by value setting, then $\termtwo\in\val$ and by means of small-step 
    semantics rules, $\ssemv{\termtwo}{\{\termtwo^{1}\}}$. 
  \item 
    Suppose $\termone=\app{\termtwo}{\termthree}$ and that we are in presence of a redex, i.e. $\termtwo$ is in the 
    form $\lambda x.\termsix$ and $\termthree$ is a value and $\ssemv{\termthree}{\{\termthree^1\}}$. 	
    Then $\mul{\termfive}$ is the unary sequence
    $\termsix\{\termthree/x\}$. The thesis easily follows by induction, taking  
    $\overline{\derone}:\ssemv{\termsix\{\termthree/x\}}{\overline{\distone}}$ as premise, where 
    $\overline{\distone}=\sup_{\ssemv{\termsix\{\termthree/x\}}{\distone_k}}^{\mbox{}}\distone_k$.
  \item 
    Suppose $\termone=\ps{\termtwo}{\termthree}$ and that $\termtwo\rv{\mul{\termsix}}$. Then 
    $\mul{\termfive}=\ps{\mul{\termsix}}{\termthree}$. 
    From the induction hypothesis applied to the derivations $\derone_i:\ssemv{\ps{\termsix_{i}}{\termthree}}{\distone_i}$, 
    we have that:
    \begin{varitemize} 
    \item 
  	Either all of the distributions $\distone_i$ is $\emdist$, then $\distone$ is $\emdist$ itself and a 
	derivation for $\ssemv{\termone}{\emdist}$ can be defined.  
      \item 
  	Or there is at least one derivation among $\derone$ to which case \ref{casefourL} applies.
	Suppose that $k_1,\ldots,k_m$ are the indices in $\interv{1}{n}$ to which case \ref{casethreeL} can
        be applied. For every in $i$ in $\{1,\ldots,m\}$, there are two
        distributions
        $\disttwo_{i}$
        and 
        $\distthree_{i}$.
        Moreover, again by induction hypothesis,  we have 
        \begin{align*}
	  \dertwo_{i}:\ssemv{\termsix_{k_i}}{\disttwo_i}\\
	  \derthree_{i}:\ssemv{\termthree}{\distthree_i}
	\end{align*}
	Moreover,  $\dertwo_i,\derthree_i\ltder\derone_{i}$ and 
  	$\distone_{k_i}\leq\frac{1}{2}\cdot 
	\disttwo_{i}\cdot(\sum\distthree_i)+\frac{1}{2} \cdot\distthree_i\cdot(\sum\disttwo_{i})$.
	Let us define $\disttwo$ as the distribution $\sum_{i=1}^{m}\frac{1}{n}\disttwo_i$
	and observe 
	that  a derivation $\dertwo:\ssemv{\termtwo}{\disttwo}$ such that $\dertwo\ltder\derone$ can 
	be defined from the derivation $\dertwo_i$. Moreover, a derivation $\derthree:\ssemv{\termthree}{\distthree}$ 
	can be defined such that  $\derthree\ltder\derone$, taking $\distthree=\sup_{i=1}^{m}\distthree_i$ 
        and applying Lemma \ref{lemma:supdist}.
	Finally we have:
    \begin{align*}
      \distone&=\sum_{i=1}^{m}\frac{1}{n}\distone_{k_{i}}\\
      &\leq \sum_{i=1}^{m}\frac{1}{n}(\frac{1}{2}\cdot \disttwo_{i}\cdot (\sumd{\distthree_i})+
        \frac{1}{2} \cdot\distthree_i\cdot(\sumd{\disttwo_{i}}))\\
      &=\sum_{i=1}^{m}\frac{1}{n}(\frac{1}{2}\cdot \disttwo_{i}\cdot(\sumd{\distthree_{i}}))+   
        \sum_{i=1}^{m}\frac{1}{n}(\frac{1}{2} \cdot\distthree_i\cdot(\sumd{\disttwo_{i}}))\\
      &\leq\sum_{i=1}^{m}\frac{1}{n}(\frac{1}{2}\cdot \disttwo_{i}\cdot(\sumd{\distthree}))+   
       \sum_{i=1}^{m}\frac{1}{n}(\frac{1}{2} \cdot\distthree\cdot(\sumd{\disttwo_{i}}))\\
      &= \frac{1}{2}\cdot \sum_{i=1}^{m}\frac{1}{n} \disttwo_{i} \cdot(\sumd{\distthree})+ 
       \frac{1}{2} \cdot\distthree\cdot ( \sum_{i=1}^{m}\frac{1}{n}(\sumd{ \disttwo_{i}}))\\
          &= \frac{1}{2}\cdot \sum_{i=1}^{m}\frac{1}{n} \disttwo_{i} \cdot(\sumd{\distthree})+ 
       \frac{1}{2} \cdot\distthree\cdot (\sumd{ ( \sum_{i=1}^{m}\frac{1}{n}\disttwo_{i}}))\\
        &=\frac{1}{2}\cdot \disttwo\cdot\sumd{\distthree}+\frac{1}{2}\cdot\distthree\cdot \sumd{\disttwo}.\\
    \end{align*}            
  \end{varitemize}          
  where for each $i\in\interv{1}{m}$, $\sumd{\distthree_i}\leq \sumd{\distthree}$ 
  holds because $\distthree$ is the least upper bound of the $\distthree_i$.           
  \item 
    Suppose $\termone=\ps{\valfive}{\termthree}$ ($\valfive\in\val$) and that $\termthree\rv{\mul{\termsix}}$. 
    Then $\mul{\termfive}=\ps{\valfive}{\mul{\termsix}}$. Similar to the previous case.
  \item 
    Suppose $\termone=\ps{\valone}{\valtwo}$ and that $\termone\rv{\valone,\valtwo}$. 
    In this case the subderivations which we are looking for are $\derone:\ssemv{\valone}{\distone}$ 
    and $\dertwo:\ssemv{\valtwo}{\disttwo}$.
\end{varitemize}
This concludes the proof.
\end{proof}

Now, suppose that $\ssemv{\termone}{\distone}$ and $\bsemv{\termone}{\disttwo}$.
Lemma~\ref{lemma:struct} provides all what is needed to ``unfold'' the hypothesis
$\ssemv{\termone}{\distone}$ and obtain judgments matching exactly those coming
from $\bsemv{\termone}{\disttwo}$. We easily get:
\begin{proposition}\label{prop:leqbs}
If $\ssemv{\termone}{\distone}$ and $\bsemv{\termone}{\disttwo}$,
then $\distone\leq\disttwo$.
\end{proposition}
\begin{proof}
By induction on the structure of a proof for $\ssemv{\termone}{\distone}$,
applying Lemma~\ref{lemma:struct} and doing some case analysis based 
on its outcome:
\begin{varitemize}
\item
  If $\distone=\emdist$, then $\distone\leq\disttwo$ trivially.
\item
  If $\termone=\valone$ and $\distone=\{\valone^1\}$, then
  $\disttwo=\distone$, because the only rule for values in the
  big step semantic is
  $$
  \czrule
      {
        \bsemv{\valone}{\{\valone\}}
      }
      {\bvv}
  $$    
\item
  If $\termone$ is an application $\app{\termtwo}{\termthree}$
  and distributions $\distthree,\distfour,\distfive_{\termfour,\valone}$ and
  derivations $\dertwo,\derthree,\derfour_{\termfour,\valone}$ 
  exist as in
  Lemma~\ref{lemma:struct}, we can observe,
  by induction hypothesis, that there exist
  distributions $\distsix$ and $\distseven$
  such that $\bsemv{\termtwo}{\distsix}$,   $\bsemv{\termthree}{\distseven}$ where
  $\distsix\geq\distthree$ and $\distseven\geq\distfour$.
  Then we have that 
  $\supp{\distsix}\supseteq \supp{\distthree}$ 
  and $\supp{\distseven}\supseteq \supp{\distfour}$. 
  Now, for every $\abstr{\varone}{\termfour}\in \supp{\distthree}$ and $\valone\in \supp{\distfour}$, suppose
  $\disteight_{\termfour,\valone}$ is such that
  $\bsemv{\subst{\termfour}{\valone}{\varone}}{\disteight_{\termfour,\valone}}$.
  Again by induction hypothesis, we obtain that
  $\disteight_{\termfour,\valone}\geq\distfive_{\termfour,\valone}$.
  Finally:
  \begin{align*}
    \distone&\leq\sum_{\abstr{\varone}{\termfour}\in\supp{\distthree}, \valone\in\supp{\distfour}}
      \distthree(\abstr{\varone}{\termfour})\cdot\distfour(\valone)\cdot\distfive_{\termfour,\valone}\\
    &\leq \sum_{\abstr{\varone}{\termfour}\in\supp{\distsix}, \valone\in\supp{\distseven}}
      \distthree(\abstr{\varone}{\termfour})\cdot\distfour(\valone)\cdot\distfive_{\termfour,\valone}\\
    &\leq \sum_{\abstr{\varone}{\termfour}\in\supp{\distsix}, \valone\in\supp{\distseven}}
      \distsix(\abstr{\varone}{\termfour})\cdot\distseven(\valone)\cdot\disteight_{\termfour,\valone}\\
      &=\disttwo.
  \end{align*}
\item
If $\termone$ is $\ps{\termtwo}{\termthree}$, we can proceed exactly as in the previous case. 
In fact, there exist distributions $\distthree, \distfour$ and derivations $\dertwo, \derthree$ 
as in Lemma~\ref{lemma:struct}, and by induction hypothesis we can observe that for any 
distributions $\distsix$ and $\distseven$ such that $\bsemv{\termtwo}{\distsix}$ and 
$\bsemv{\termthree}{\distseven}$,  $\distsix\geq\distthree$ and $\distseven\geq\distfour$ hold.
  
We can take  $\distsix$ such that $\supp{\distsix}\supseteq \supp{\distthree}$ and 
$\distsix(\valone)\geq\distthree(\valone)$ for each $\valone\in\supp{\distthree}$, and we 
can take $\distseven$ such that $\supp{\distseven}\supseteq \supp{\distfour}$ and 
$\distseven(\valone)\geq\distfour(\valone)$ for each $\valone\in\supp{\distfour}$. Then we have:
\begin{align*}
  \distone&\leq\frac{1}{2}\cdot\distthree\cdot(\sum \distfour)+\frac{1}{2}\cdot\distfour\cdot(\sum \distthree)\\
  &\leq\frac{1}{2}\cdot\distsix\cdot(\sum \distseven)+\frac{1}{2}\cdot\distseven\cdot(\sum \distsix)\\
  &=\disttwo.
\end{align*}
\end{varitemize}
This concludes the proof.
\end{proof}

\subsubsection{Small-Step is in Big-Step}\label{sect:ssinbs}
Whenever $\ssemv{\termone}{\distone}$ and $\termone$ is not a value, one can always
``decompose'' $\distone$ and find some judgments about 
the immediate subterms of $\termone$. This is Lemma~\ref{lemma:struct}.
If we want to prove that \emph{the} small-step semantics $\Ssemv{\termone}$ of $\termone$
can be itself attributed to $\termone$ in the big-step case, we should somehow prove the
converse, namely that judgments about the immediate subterms of $\termone$ can be packaged
into an analogous judgment for $\termone$. 
\begin{lemma}\label{lemma:structreverse}
Let $\termone\in\LOP$ be any term. Then: 
\begin{varenumerate}
\item\label{case1sb1} 
  If $\termone$ is a value $\valone$ and $\ssemv{\termone}{\distone}$, then $\distone\leq\{\valone^1\}$.
\item\label{case2sb1} 
  If $\termone$ is an application $\app{\termtwo}{\termthree}$, 
  $\ssemv{\termtwo}{\disttwo}$, $\ssemv{\termthree}{\distthree}$, 
  $\{\ssemv{\subst{\termfour}{\valone}{\varone}}{\distfour_{\termfour,\valone}}\}_{\abstr{\varone}{\termfour}\in\supp{\disttwo}, \valone\in\supp{\distthree}}$ 
  then there exist a distribution $\distone$ such that $\ssemv{\termone}{\distone}$ and 
  $\distone\geq\sum_{\abstr{\varone}{\termfour}\in\supp{\disttwo}, 
  \valone\in\supp{\distthree}}\disttwo(\abstr{\varone}{\termfour})\cdot\distthree(\valone)\cdot\distfour_{\termfour, \valone}$.
\item \label{case3sb1} 
  If $\termone$ is a sum $\ps{\termtwo}{\termthree}$, $\ssemv{\termtwo}{\disttwo}$, 
  $\ssemv{\termthree}{\distthree}$ then there exist a distribution $\distone$ such that $\ssemv{\termone}{\distone}$ and 
  $\distone\geq\frac{1}{2}\cdot\disttwo\cdot\sumd{\distthree}+\frac{1}{2}\cdot\distthree\cdot\sumd{\disttwo}$.
\end{varenumerate}
\end{lemma}
\begin{proof}
Let us prove the three statements separately:
\begin{varitemize}
\item
  If $\termone$ is a value $\valone$, then the only possible judgments involving $\termone$
  are $\ssemv{\termone}{\{\valone^1\}}$ and $\ssemv{\termone}{\emdist}$. The thesis trivially holds.
\item 
  If $\termone$ is an application $\app{\termtwo}{\termthree}$, then
  we prove statement~\ref{case2sb1} by induction on the derivations 
  $\ssemv{\termtwo}{\disttwo}$ and $\ssemv{\termthree}{\distthree}$. Let's distinguish some cases:
  \begin{varitemize}
  \item  
    If $\ssemv{\termtwo}{\emdist}$, then 
    $\sum_{}\disttwo(\abstr{\varone}{\termfour})\cdot\distthree({\valone})\cdot\distfour_{\termfour,\valone}=\emdist$. 
    We derive $\ssemv{\termone}{\emdist}$ by means of the rule $\sev$ and the thesis holds.
  \item 
    If $\ssemv{\termtwo}{\disttwo}\neq\emdist$, suppose that the least rule applied in the derivation is 
    $$\brule
    {
      \termtwo\rv\mul{\termsix}
    }
    {
      \ssemv{\termsix_i}{\disttwo_i}
    }
    {
      \ssemv{\termtwo}{\sum_{i=1}^n\frac{1}{n}\disttwo_i}
    }
    {}
    $$
    Note that $\disttwo=\sum_{i=1}^{n}\frac{1}{n}\disttwo_{i}$.
    It is possible to apply the induction hypothesis on every ${\app{\termsix_{i}}{\termthree}}$:
    if  $\ssemv{\termsix_{i}}{\disttwo_i}$, $\ssemv{\termthree}{\distthree}$ 
    and $\ssemv{\subst{\termfour}{\valone}{\varone}}{\distfour_{\termfour,\valone}$ 
    (for $\abstr{\varone}{\termfour}\in\supp{\disttwo_i}\subseteq\supp{\disttwo},\valone\in\supp{\distthree}}$) 
    follows that  $\ssemv{\app{\termsix_{i}}{\termthree}}{\distfive_i}$ 
    with $\distfive_i \geq\sum_{\abstr{\varone}{\termfour}\in\supp{\disttwo_i},\valone\in\supp{\distthree}}
    \disttwo_i(\abstr{\varone}{\termfour})\distthree(\valone)\distfour_{\termfour,\valone}$.
    We are able to construct a derivation  $\distone$ in the following way:
    $$
    \brule
        {
          \app{\termtwo}{\termthree}\rv{\app{\mul{\termsix}}{\termthree}}
        }
        {
          \ssemv{\app{\termsix_i}\termthree}{\distfive_i}
        }
        {
          \ssemv{\app{\termtwo}{\termthree}}{\sum_{i=1}^{n}\frac{1}{n}\distfive_i}
        }
        {}
        $$
    And finally we have:
   \begin{align*}
    \distone&=\sum_{i=1}^{n}\frac{1}{n}\distfive_{i}\geq \sum_{i=1}^{n}\frac{1}{n}
       \left( \sum_{\abstr{\varone}{\termfour}\in\supp{\disttwo_i},\valone\in{\supp{\distthree}}}
       \disttwo_i(\abstr{\varone}{\termfour})\distthree(\valone)\distfour_{\termfour,\valone}\right)\\
    &=\sum_{i=1}^{n}\frac{1}{n}
       \left( \sum_{\abstr{\varone}{\termfour}\in\supp{\disttwo},\valone\in{\supp{\distthree}}}
       \disttwo_i(\abstr{\varone}{\termfour})\distthree(\valone)\distfour_{\termfour,\valone}\right)\\
    &=\sum_{\abstr{\varone}{\termfour}\in\supp{\disttwo},\valone\in{\supp{\distthree_i}}}
       \left(\sum_{i=1}^{n}\frac{1}{n}\left(\disttwo_i(\abstr{\varone}{\termfour})\right)\distthree_i(\valone)
       \distfour_{\termfour,\valone}\right)\\
    &=\sum_{\abstr{\varone}{\termfour}\in\supp{\disttwo},\valone\in{\supp{\distthree}}}\disttwo
       (\abstr{\varone}{\termfour})\distthree(\valone)\distfour_{\termfour,\valone}
   \end{align*}
\item 
  Other cases follows can be handled similarly to the previous one.
\end{varitemize}
\item 
  $\termone$ is a sum $\ps{\termtwo}{\termthree}$. 
  We prove the result by induction on the derivations 
  $\ssemv{\termtwo}{\disttwo}$ and $\ssemv{\termthree}{\distthree}$.
  \begin{varitemize}
  \item 
    If $\ssemv{\termtwo}{\emdist}$, then $\frac{1}{2}\cdot\emdist\cdot\sumd{\distthree}+\frac{1}{2}
    \cdot\distthree\cdot\sumd{\emdist}=\emdist$. We derive $\ssemv{\termone}{\emdist}$ by means of 
    the rule $\sev$ and the thesis holds.
  \item 
    If $\termtwo$ is a value $\valone$, the only interesting case is the one in which 
    $\termthree$ is not a value and also $\ssemv{\termthree}{\distthree}$ with $\distthree\neq\emdist$ 
    (in fact if $\ssemv{\termthree}{\emdist}$ we are in the previous case and if $\termthree$ 
    is a value the proof is trivial).
    Then $M=\ps{\valone}{\termthree}$.
    Let us consider the derivation of the judgment $\ssemv{\termthree}{\distthree}$: 
    the last rule applied in the derivation is
    $$
    \brule
        {
          \termthree\rv\mul{\termsix}
        }
        {
          \ssemv{\termsix_i}{\distthree_i}
        }
        {
          \ssemv{\termthree}{\sum_{i=1}^n\frac{1}{n}\distthree_i}
        }
        {}
    $$
    where $\distthree=\sum_{i=1}^n\frac{1}{n}\distthree_i$.
    Observe that $\termone\rv{\ps{\valone}{\mul{\termsix}}}$.
    We can apply the induction hypothesis to each $\ssemv{\termthree}{\sum_{i=1}^n\frac{1}{n}\distthree_i}$ and
    to $\ssemv{\valone}{\{\valone^1\}}$, and we obtain
    $$
    \distone_i\geq\frac{1}{2}\cdot\{\valone^1\}\cdot\sumd{\distthree_i}+\frac{1}{2}\cdot\distthree_i.
    $$
    We are able to construct a derivation in the following way:
    $$
    \brule
        {
          \ps{\valone}{\termthree}\rv{\ps{\valone}{\mul{\termsix}}}
        }
        {
          \ssemv{\ps{\valone}{\termsix_i}}{\distone_i}
        }
        {
          \ssemv{\termone}{\sum_{i=1}^{n}\frac{1}{n}\distone_i}
        }
        {}
    $$
    Finally we have:
   \begin{align*}
    \sum_{i=1}^{n}\frac{1}{n}\distone_{i}&\geq \sum_{i=1}^{n}\frac{1}{n}\left(\frac{1}{2}
      \cdot\{\valone^1\}\cdot\sumd{\distthree_i}+\frac{1}{2}\cdot\distthree_i\right)\\
    &=\sum_{i=1}^{n}\frac{1}{n}\left(\frac{1}{2}\cdot\{\valone^1\}\cdot\sumd{\distthree_i}\right)+
      \sum_{i=1}^{n}\frac{1}{n}\left(\frac{1}{2}\cdot\distthree_i\right)\\
    &=\frac{1}{2}\cdot\{\valone^1\}\cdot \left(\sumd{\sum_{i=1}^{n}\frac{1}{n}\distthree_i}\right)+
      \frac{1}{2}\cdot\sum_{i=1}^{n}\frac{1}{n}\distthree_i\\
    &=\frac{1}{2}\cdot\{\valone^1\}\cdot \left(\sumd{\distthree}\right)+\frac{1}{2}\cdot\distthree\\
   \end{align*}
 \item
   If $\ssemv{\termtwo}{\disttwo}$, suppose that the least rule applied in the derivation is 
   $$
   \brule
    {
      \termtwo\rv\mul{\termsix}
    }
    {
      \ssemv{\termsix_i}{\disttwo_i}
    }
    {
      \ssemv{\termtwo}{\sum_{i=1}^n\frac{1}{n}\disttwo_i}
    }
    {}
   $$
    Note that $\disttwo=\sum_{i=1}^{n}\frac{1}{n}\disttwo_{i}$.  
    It is possible to apply the induction hypothesis on the single ${\ps{\termsix_{i}}{\termthree}}$; 
    then, if  $\ssemv{\termsix_{i}}{\disttwo_i}$, $\ssemv{\termthree}{\distthree}$  we have that
    $\ssemv{\ps{\termsix_{i}}{\termthree}}{\distone_i}$ with 
    $\distone_i \geq\frac{1}{2}\cdot\disttwo_i\cdot\sumd{\distthree_i}+ \frac{1}{2}\cdot\distthree_i\cdot\sumd{\disttwo_i}$.
    We  construct a derivation in the following way:
    $$
    \brule
        {
          \ps{\termtwo}{\termthree}\rv{\ps{\mul{\termsix}}{\termthree}}
        }
        {
          \ssemv{\ps{\termsix_i}{\termthree}}{\distone_i}
        }
        {
          \ssemv{\termone}{\sum_{i=1}^{n}\frac{1}{n}\distone_i}
        }
        {}
     $$
     Finally we have:
     \begin{align*}
       \sum_{i=1}^{n}\frac{1}{n}\distone_{i}&\geq \sum_{i=1}^{n}\frac{1}{n}\left(\frac{1}{2}
       \cdot\disttwo_i\cdot\sumd{\distthree_i}+\frac{1}{2}\cdot\distthree_i\cdot{\sumd{\disttwo_i}}\right)\\
       &= \sum_{i=1}^{n}\frac{1}{n}\left(\frac{1}{2}\cdot\disttwo_i\cdot\sumd{\distthree_i}\right)+
       \sum_{i=1}^{n}\frac{1}{n}\left(\frac{1}{2}\cdot\distthree_i\cdot{\sumd{\disttwo_i}}\right)\\
       &= \frac{1}{2}\cdot\sum_{i=1}^{n}\frac{1}{n}\left(\disttwo_i\cdot\sumd{\distthree_i}\right)+
       \frac{1}{2}\sum_{i=1}^{n}\frac{1}{n}\left(\distthree_i\cdot{\sumd{\disttwo_i}}\right)\\
       &= \frac{1}{2}\cdot\disttwo\cdot\sumd{\distthree}+\frac{1}{2}\cdot\distthree\cdot{\sumd{\disttwo}}\\
     \end{align*}
\item 
  Other cases are trivial.
\end{varitemize}
\end{varitemize}
This concludes the proof.
\end{proof}

Lemma~\ref{lemma:struct} and Lemma~\ref{lemma:structreverse} allows to prove that the
$\Ssemv{\cdot}$ commutes well with the various constructs of $\LOP$.
\begin{lemma}\label{lemma:sb2}
 For each term $\termone\in\LOP$, then
\begin{varenumerate}
\item\label{case1sb2} 
  If $\termone$ is a value $\valone$, then $\Ssemv{\termone}=\{\valone^1\}$;
\item\label{case2sb2} 
  If $\termone$ is an application $\app{\termtwo}{\termthree}$, then 
  $\Ssemv{\termone}=\sum_{\abstr{\varone}{\termfour}\in\supp{\Ssemv{\termtwo}}, \valone
  \in\supp{\Ssemv{\termthree}}}\Ssemv{\termtwo}(\abstr{\varone}{\termfour})\cdot
  \Ssemv{\termthree}(\valone)\cdot\Ssemv{\subst{\termfour}{\valone}{\varone}}$;
\item\label{case3sb2} 
  If $\termone$ is a sum $\ps{\termtwo}{\termthree}$, then
  $\Ssemv{\termone}=\frac{1}{2}\cdot\Ssemv{\termtwo}\cdot\sumd{\Ssemv{\termthree}}+
  \frac{1}{2}\cdot \Ssemv{\termthree}\cdot\sumd{\Ssemv{\termtwo}}$.
\end{varenumerate}
\end{lemma}
\begin{proof}
We will use the following fact throughout the proof:
\begin{fact}\label{fact:one}
If $\termone\rv\mul{\termtwo}$, then $\sup_{\ssemv{\termone}{\distone}}\distone=\sup_{\ssemv{\termtwo_i}{\distone_i}}(\sum_{i=1}^{n}\frac{1}{n}\distone_i)$.
\end{fact}
The inequalities above can be proved separately:
\begin{varitemize}
\item
  If $\termone$ is a value, the thesis follows by small step semantics rules. Indeed,
  $\ssemv{\termone}{\emdist}$. 
\item
  For the other cases, the 
  $\mathbf{(\leq)}$ 
  direction follows from Lemma~\ref{lemma:struct} and Fact~\ref{fact:one};
  $\geq$ direction follows from Lemma~\ref{lemma:structreverse}  and Fact~\ref{fact:one}.
\end{varitemize}
This concludes the proof.
\end{proof}

The fact $\Ssemv{\termone}$ can be assigned to $\termone$ in the big-step
semantics is an easy consequence of Lemma~\ref{lemma:sb2}:
\begin{proposition}\label{prop:supbsem}
$\bsemv{\termone}{\Ssemv{\termone}}$.
\end{proposition}
\begin{proof} 
We will prove the thesis by coinduction:
We can prove that all judgments $\bsemv{\termone}{\Ssemv{\termone}}$
belong to the coinductive interpretation of the underlying formal system $\isone$ 
(in this case, the formal system is $\isone=\{{\bvv}, {\bav}, {\bsv}\}$).
To do that, we need to prove that the set $\sjone$ of all those judgments
is consistent, i.e. that $\sjone\subseteq\opinf{\isone}(\sjone)$.
This amounts to show that if $\juone\in\sjone$, then there is a derivation
for $\juone$ whose immediate premises are themselves in $\sjone$. Let's
distinguish some cases:
\begin{varitemize}
\item 
  If $\termone=\valone$ then $\bsemv{\valone}{\{V^{1}\}}$ by $\mathsf{bv_v}$ rule, 
  and $\{V^{1}\}=\Ssemv{\valone}$ because of Lemma~\ref{lemma:sb2}.
\item 
  If $\termone$ is an application $\app{\termtwo}{\termthree}$, take the 
  judgment $\juone_1=\bsemv{\termtwo}{\Ssemv{\termtwo}}$, 
  $\juone_2=\bsemv{\termthree}{\Ssemv{\termthree}}$ and the family of judgments 
  $\{\bsemv{\{\subst{\termfour}{\valone}{\varone}\}}
  {\Ssemv{\{\subst{\termfour}{\valone}{\varone}\}}}\}_{\abstr{\varone}{\termfour}\in\supp{\Ssemv{\termtwo}}, \valone\in\supp{\Ssemv{\termthree}}}$: 
  we will prove that the judgment $\bsemv{\app{\termtwo}{\termthree}}{\Ssemv{\app{\termtwo}{\termthree}}}$ can be derived in a 
  single step from $\juone_1$, $\juone_2$ and those in the family above
  by means of $\mathsf{ba_v}$ rule.
  Simply observe that
  $$
  \trule
      {\bsemv{\termtwo}{\Ssemv{\termtwo}}}
      {\bsemv{\termthree}{\Ssemv{\termthree}}}
      {\{\bsemv{\{\subst{\termfour}{\valone}{\varone}\}}
  {\Ssemv{\{\subst{\termfour}{\valone}{\varone}\}}}\}_{\abstr{\varone}{\termfour}\in\supp{\Ssemv{\termtwo}}, \valone\in\supp{\Ssemv{\termthree}}}}
      {\bsemv{\app{\termtwo}{\termthree}}{\sum_{\abstr{\varone}{\termfour}\in\supp{\Ssemv{\termtwo}}, \valone\in\supp{\Ssemv{\termthree}}}\Ssemv{\termtwo}(\abstr{\varone}{\termfour})\cdot\Ssemv{\termthree}(\valone)\cdot\Ssemv{\subst{\termfour}{\valone}{\varone}}}}
      {{\bav}}
  $$ 
  The thesis follows applying Lemma~\ref{lemma:sb2}.
\item 
  If $\termone$ is a sum $\ps{\termtwo}{\termthree}$, take the judgment $c_1=\bsemv{\termtwo}{\Ssemv{\termtwo}}$ 
  and $\juone_2=\bsemv{\termone}{\Ssemv{\termthree}}$: we will prove that the judgment 
  $\bsemv{\ps{\termtwo}{\termthree}}{\Ssemv{\app{\termtwo}{\termthree}}}$ can be inferred in a single step from 
  $\juone_1$ and $\juone_2$ by means of $\mathsf{bs_v}$ rules.
  Clearly, $\juone_1$ and $\juone_2$ belong to $\sjone$. Moreover
  $$
  \brule
      {\bsemv{\termtwo}{\Ssemv{\termtwo}}}
      {\bsemv{\termthree}{\Ssemv{\termthree}}}
      {\bsemv{\ps{\termtwo}{\termthree}}{\frac{1}{2}\cdot\Ssemv{\termtwo}\cdot\sumd{\Ssemv{\termthree}}}+\frac{1}{2}\cdot\Ssemv{\termthree}\cdot\sumd{\Ssemv{\termtwo}}}
      {{\bsv}}
  $$ 
  and by  Lemma~\ref{lemma:sb2}, case~\ref{case3sb2} we obtain the thesis. 
\end{varitemize}
This concludes the proof.
\end{proof}

The equality between big-step and small-step semantics is a corollary
of Proposition~\ref{prop:leqbs} and Proposition~\ref{prop:supbsem}.
\begin{theorem}\label{th:cbvSB}
$\Bsemv{\termone}=\Ssemv{\termone}$.
\end{theorem}

\section{Call-by-Name}\label{sec:cbn}
In Section~\ref{sec:cbv} we endowed $\LOP$ with a \emph{call-by-value} probabilistic operational semantics and showed
that the distribution assigned to any term $\termone$ is the same in big-step and in small-step semantics, independently
on whether they are defined inductively or coinductively. Actually, the same holds in call-by-name: both big-step and 
small-step semantics can be defined (co)inducively and proved equivalent, following the same path used in call-by-value. 
In this section, we briefly sketch how this can be done, by defining
$\Bsemn{\termone}$ and $\Ssemn{\termone}$ and by proving they are equivalent.

\begin{definition}[Call-by-name Reduction]\label{def:leftRedN}
Leftmost reduction $\rn$ is the least binary relation between $\LOP$ and $\LOP^*$
such that:
\begin{align*}
(\abstr{\varone}{\termone}){\termtwo}&\rn\subst{\termone}{\varone}{\termtwo}&
\termone\termtwo&\rn\mul{\termthree}\termtwo\qquad\mbox{if $\termone\rn\mul{\termthree}$}\\
\ps{\termone}{\termtwo}&\rn\termone,\termtwo & &
\end{align*}
\end{definition}
Note that, contrarily to call-by-value, in call-by-name it \emph{is} possible to perform a choice between terms which
are not values. 

\subsection{Small-step Semantics} 
As for call-by-name, we model separately terminating and non-terminating computations. The rule 
schema is the same, up to the different reduction relation $\rn$. First of all, there is an inductively defined 
binary relation $\ssemn{}{}$ between $\LOP$ and distributions. Rules are as
follows:
$$
\zrule
    {
      \ssemn{\termone}{\emdist}
    }
    {{\sen}}
\qquad
\zrule
    {
      \ssemn{\valone}{\{\valone^1\}}
    }
    {{\svn}}
\qquad
\brule
    {
      \termone\rn\mul{\termtwo}
    }
    {
      \ssemn{\termtwo_i}{\distone_i}
    }
    {
      \ssemn{\termone}{\sum_{i=1}^n\frac{1}{n}\distone_i}
    }
    {{\smn}}
$$
$\Ssemn{\termone}$ is the distribution $\sup_{\ssemn{\termone}{\distone}}\distone$. Moreover,
there is also another, coinductively defined binary relation between 
$\LOP$ and $\RRp{0}{1}$ capturing divergence. Rules are as follows:
$$
\czrule
    {
      \sdsemn{\valone}{0}
    }
    {{\dvn}}
\qquad
\cbrule
    {
      \termone\rn\mul{\termtwo}
    }
    {
    \{  \sdsemn{\termtwo_i}{\probone_i}\}_{i\in\interv{1}{n}}
    }
    {
      \sdsemn{\termone}{\sum_{i=1}^n\frac{1}{n}\probone_i}
    }
    {{\dmn}}
$$
$\Sdsemn{\termone}$ is nothing more than $\sup_{\sdsemn{\termone}{\probone}}\probone$.
Notice how the differences between call-by-value and call-by-name small-step semantics all come from
the reduction relation, since the rules above are analogous to their call-by-value siblings.

As done in Section~\ref{sec:cobsv} for call-by-name,  a conductive version of call-by-name small step semantics can be easily defined.

\subsection{Big-step Semantics.}
We define call-by-name big-step semantics of terms in $\LOP$ as the co-inductive interpretation of a suitable set of rules.
Again, this allow us to capture infinite computations. A coinductively defined binary relation $\bsemn{}{}$ between 
$\LOP$ and distributions is obtained by taking all instances of the following rules:
$$
\czrule
    {
      \bsemn{\valone}{\{\valone\}}
    }
    {\bvn}
\qquad
\cbrule
    { 
      \bsemn{\termone}{\distone}
    }
    {
     \{ \bsemn{\subst{\termfour}{\termtwo}{\varone}}{\disttwo_{\termfour,\termtwo}}\}_{\abstr{\varone}{\termfour}\in\supp{\distone}}
    }
    {
      \bsemn{\app{\termone}{\termtwo}}{\sum_{\abstr{\varone}{\termfour}\in\supp{\distone}}\distone(\abstr{\varone}{\termfour})\cdot\disttwo_{\termfour,\termtwo}}
    }
    {{\ban}
    }
\qquad
\cbrule
    {
      \bsemn{\termone}{\distone}
    }
    {
      \bsemn{\termtwo}{\disttwo}
    }
    {
      \bsemn{\ps{\termone}{\termtwo}}{\frac{1}{2}\cdot\distone+\frac{1}{2}\cdot\disttwo}
    }
    {{\bsn}
    }
$$
$\Bsemn{\termone}$ is simply the subdistribution $\inf_{\bsemn{\termone}{\distone}}\distone$. The way
binary choices are managed reflects the reduction rules, which allow to evaluate a binary choice
to one of its components even if the latter are not values. Indeed, while in call-by-value 
normalization factors $\sumd{\distone}$ and $\sumd{\disttwo}$ were necessary, they are not here
anymore.

\begin{example}
Consider the term $\termone=\ps{\Omega}{(\abstr{\varone}{\varone})}$. Recall from 
Example~\ref{exa:omega}, that $\bsemv{\Omega}{\distone}$ for every $\distone$.
Moreover, $\bsemv{\abstr{\varone}{\varone}}{\{(\abstr{\varone}{\varone})^1\}}$.
This implies $\bsemv{\termone}{\emdist}$, and, as a consequence, that
$\Bsemv{\termone}=\emdist$. The same behavior cannot be mimicked in 
call-by-name. Indeed, while $\bsemn{\Omega}{\distone}$ (for every $\distone$) and
$\bsemn{\abstr{\varone}{\varone}}{\{(\abstr{\varone}{\varone})^1\}}$, the
smallest distribution which can be assigned to $\termone$ is
$\Bsemn{\termone}=\{(\abstr{\varone}{\varone})^{\frac{1}{2}}\}$.
\end{example}

Inductive call-by-name big step semantics can be obtained adding the rule which assigns the 
empty distribution to any term and taking the inductive interpretation of the system.

\subsection{Comparing the Different Notions}
We retrace here results of Section \ref{sec:cbvDvsC} and Section~\ref{sec:cbvSvsB} for call-by-name semantics.

\subsubsection{Divergence and Convergence in CbN Small-step Semantics}\label{sec:cbnDvsC}

Convergence and divergence call-by-name small step semantics are proved equivalent by means of the following results.

\begin{proposition}
If $\sdsemn{\termone}{\probone}$ and $\ssemn{\termone}{\distone}$, then 
$\probone+\sumd{\distone}\leq 1$.
\end{proposition}
\begin{proof}
The proof goes by induction on the structure of the derivation
for $\ssemn{\termone}{\distone}$.\end{proof}

\begin{theorem}
For every term $\termone$, $\sdsemn{\termone}{1-\sumd{\Ssemn{\termone}}}$.
\end{theorem}
\begin{proof} By coinduction: it is possible to  prove that all judgments $\sdsemn{\termone}{1-\sumd{\Ssemn{\termone}}}$
belong to the coinductive interpretation of the underlying formal system $\isone$.
\end{proof}

\subsubsection{Small-step vs. Big-step}\label{sec:cbnSvsB}

We prove here the equivalence between inductive call-by-name small step semantics $\Ssemn{\termone}$ 
and coinductive call-by-name big step semantics $\Bsemn{\termone}$.
We use the same proof technique of Section~\ref{sec:cbvDvsC}.

Given any term $\termone$, the sets of distributions
sets of probability distributions which can be attributed to $\termone$
in call-by-name small-step semantics and call-by-name big-step semantics, are respectively defined as 
$\Ssemns{\termone}=\{\distone\mid\ssemn{\termone}{\distone}\}$ and 
$ \Bsemns{\termone}=\{\distone\mid\bsemn{\termone}{\distone}\}$.
Firstly we prove that big-step semantics \emph{dominates} small-step
semantics (Proposition~\ref{prop:leqbsn}), 
and this imply that  $\Bsemn{\termone}\geq\Ssemn{\termone}$. Successively, we prove again that small step semantics can itself be derived using
big-step semantics, namely that $\Ssemn{\termone}\in\Bsemns{\termone}$ (Proposition~\ref{prop:supbsemn}). 
As a consequence,  $\Bsemn{\termone}\leq\Ssemn{\termone}$ and finally $\Ssemns{\termone}=\Bsemns{\termone}$. 

\begin{lemma}\label{lemma:supdistn}
  If $I$ is finite and $\derone_i:\ssemn{\termone}{\distone_i}$ for every $i\in I$, then
  there is $\dertwo:\ssemn{\termone}{\sup_{i\in I}\distone_i}$. Moreover
  $\dertwo\leqder\derone_i$ for some $i\in I$.
\end{lemma}
\begin{proof}
The proof goes by induction on the structure of the proofs
in the family $\{\derone_i\}_{i\in I}$ (which can actually be
done, since $I$ is finite).
\end{proof}

\begin{lemma}\label{lemma:structn}
If $\ssemn{\termone}{\distone}$, then at least
one of the following conditions hold:
\begin{varenumerate}
\item\label{caseoneN}
  $\distone=\emptyset$;
\item\label{casetwoN}
  $\termone$ is a value $\valone$ and $\distone=\{\valone^1\}$;
\item\label{casethreeN}
  $\termone$ is an application $\app{\termtwo}{\termthree}$ and
  there is a (finite) distribution $\disttwo$
 and 
  for every $\abstr{\varone}{\termfour}\in\supp{\disttwo} $ a distribution
  $\distfour_{\termfour,\termthree}$ 
  such that:
  \begin{varenumerate}
  \item\label{case3oneN}
    $\dertwo:\ssemn{\termtwo}{\disttwo}$
    and $\derfour_{\termfour, \termthree}:\ssemn{\subst{\termfour}{\varone}{\termthree}}{\distfour_{\termfour, \termthree}}$;
  \item\label{case3twoN}
    $\dertwo,\derfour_{\termfour, \termthree}\ltder\derone$;
  \item\label{case3threeN}
    $\distone\leq\sum_{\abstr{\varone}{\termfour}\in\supp{\disttwo}}\disttwo(\abstr{\varone}{\termfour})\cdot\distfour_{\termfour, \termthree}$.
  \end{varenumerate}
\item\label{casefourN}
  $\termone$ is a sum $\ps{\termtwo}{\termthree}$ and there are (finite) distributions
  $\disttwo$ and
  $\distthree$ such that:
  \begin{varenumerate}
  \item\label{case4oneN}
    $\dertwo:\ssemn{\termtwo}{\disttwo}$ and $\derthree:\ssemn{\termthree}{\distthree}$;
  \item\label{case4twoN}
    $\dertwo,\derthree\ltder\derone$;
  \item\label{case4threeN}
    $\distone\leq \frac{1}{2}\cdot\disttwo\cdot+
    \frac{1}{2}\cdot\distthree$.   
  \end{varenumerate}
\end{varenumerate}
\end{lemma}
\begin{proof}
This is an induction on derivations for $\ssemv{\termone}{\distone}$.
The cases where $\derone$ is obtained by the rules without premises
are trivial. So, we can assume that $\derone$ has the form
$$
\brule
    {
      \termone\rv\mul{\termfive}
    }
    {
   \{   \derone_t:\ssemn{\termfive_i}{\distone_i}\}_{i\in\interv{1}{n}}
    }
    {
      \ssemn{\termone}{\sum_{i=1}^n\frac{1}{n}\distone_i}
    }
    {\smv}
$$
where $\mul{\termfive}$ is $\termfive_1,\ldots,\termfive_n$.
Let's distinguish some cases depending on
how the premise $\termone\rv\mul{\termfive}$ is derived
\begin{varitemize}
  \item
  
  Suppose $\termone=\app{\termtwo}{\termthree}$ and
  that $\termtwo\rv\mul{\termsix}$. Then
  $\mul{\termfive}=\app{\mul{\termsix}}{\termthree}$. From 
  the induction hypothesis applied to
  the derivations $\derone_i:\ssemn{\app{\termsix_i}{\termthree}}{\distone_i}$, we obtain that
  \begin{varitemize}
    \item
      Either all of the $\distone_i$ is $\emdist$. In this case $\distone$ is itself $\emdist$, and a 
      derivation for $\ssemn{\termone}{\emdist}$ can be defined.
    \item
      Or there is at least one among the derivation $\derone_i$ to which case \ref{casethreeN}\ applies.
      Suppose that $k_1,\ldots,k_m$ are the indices in $\interv{1}{n}$ to which case \ref{casethreeN} can
      be applied. This means that for every in $i$ in $\{1,\ldots,m\}$, there is
      a distribution $\disttwo_{i}$, 
      and for every $\abstr{\varone}{\termfour}\in{\supp{\disttwo_i}}$  
      a distribution
      $\distfive^{i}_{\termfour^{i}}$ such that: 
      \begin{align*}
         \dertwo_i&: \ssemv{\termsix_{k_i}}{\disttwo_{i}};\\ 
         \derfive^{i}_{\termfour}&: \ssemv{\termfour\{\termthree/\varone\}}{\distfive^{i}_{\termfour}}.
      \end{align*}
      Moreover, $ \dertwo_i,\derfive^{i}_{g,h}\ltder\derone_{t}$,  and 
      $$
      \distone_{k_i}\leq\sum_{\abstr{\varone}{\termfour^{}}\in{\supp{\disttwo_i}}}\disttwo_{i}(\termfour^{})\cdot\distfive^{i}_{\termfour}.
      $$
      Now, define $\disttwo$ as the distribution
      $$
      \sum_{i=1}^m\frac{1}{n}\disttwo_i
      $$
      
      Note that $\supp{\disttwo}=\bigcup_{i=1}^{m}\supp{\disttwo_i}$. 
      Clearly, a derivation $\dertwo$ for $\ssemv{\termtwo}{\disttwo}$ can be defined such that $\dertwo\ltder\derone$: simply
      construct it from the derivations $\dertwo_i$ and some derivations for $\ssemv{\termsix_j}{\emdist}$.
      Moreover,
      a derivations $\derfour_{\termfour}$ can be defined for every
      $\abstr{\varone}{\termfour}\in\supp{\disttwo}$ 
      in such a way that $\derfour_{\termfour}\ltder\derone$ (use Lemma \ref{lemma:supdistn}),
      $\derfour_{\termfour}:\ssemv{\termfour\{\termthree/\varone\}}{\distfour_{\termfour}}$
      and $\distfour_{\termfour}=\sup_{i=1}^m\distfive_{\termfour}^i$.
      Now:
      \begin{align*}
          \distone&=\sum_{i=1}^{n}\frac{1}{n}\distone_{i}\leq \sum_{i=1}^{m}\frac{1}{n}\distone_{k_i}\\
          &\leq\sum_{i=1}^{m}\frac{1}{n} \left(\sum_{\abstr{\varone}{\termfour}\in\supp{\disttwo_i}}\disttwo_{i}(\abstr{\varone}{\termfour}) \cdot\distfive^{i}_{\termfour}\right)\\
          &\leq\sum_{i=1}^{m}\frac{1}{n} \left(\sum_{\abstr{\varone}{\termfour}\in\supp{\disttwo},
               }\disttwo_{i}(\abstr{\varone}{\termfour})\cdot\distfive^{i}_{\termfour}\right)\\
          &=\sum_{\abstr{\varone}{\termfour}\in\supp{\disttwo}
              }\sum_{i=1}^{m}\frac{1}{n}
               \left(\disttwo_{i}(\abstr{\varone}{\termfour})\cdot
             \distfive^{i}_{\termfour}\right)\\
          &\leq\sum_{\abstr{\varone}{\termfour}\in\supp{\disttwo_i}}
               \sum_{i=1}^{m}\frac{1}{n}\left(\disttwo_{i}(\abstr{\varone}{\termfour})
               \cdot\distfour_{\termfour}\right)\\
          &=\sum_{\abstr{\varone}{\termfour}\in\supp{\disttwo}}
               \disttwo(\abstr{\varone}{\termfour})\cdot\distfour^{}_{\termfour}.
      \end{align*}             
  \end{varitemize}
  \item 
    Suppose $\termone=\app{\termtwo}{\termthree}$ and that we are in presence of a redex, i.e. $\termtwo$ is in the 
    form $\abstr{\varone}{\termsix}$. Then $\disttwo$ is simply $\{\abstr{\varone}{\termsix}^1\}$.
    From $\ssemv{\termsix\{\termthree/x\}}{\distone}$, the thesis easily follows.
  \item 
    Suppose $\termone=\ps{\termtwo}{\termthree}$ and that $\termone\rn{\termtwo, \termthree}$. 
    In this case the subderivations which we are looking for are $\derone:\ssemn{\termtwo}{\distone}$ 
    and $\dertwo:\ssemn{\termthree}{\disttwo}$.
\end{varitemize}
This concludes the proof.
\end{proof}

\begin{proposition}\label{prop:leqbsn}
If $\ssemn{\termone}{\distone}$ and $\bsemn{\termone}{\disttwo}$,
then $\distone\leq\disttwo$.
\end{proposition}
\begin{proof}
By induction on the structure of a proof for $\ssemn{\termone}{\distone}$,
applying Lemma~\ref{lemma:structn} and doing some case analysis based 
on its outcome:
\begin{varitemize}
\item
  If $\distone=\emdist$, then $\distone\leq\disttwo$ trivially.
\item
  If $\termone=\valone$ and $\distone=\{\valone^1\}$, then
  $\disttwo=\distone$, because the only rule for values in the
  big step semantic is
  $$
  \czrule
      {
        \bsemn{\valone}{\{\valone\}}
      }
      {\bvn}
  $$    
\item
  If $\termone$ is an application $\app{\termtwo}{\termthree}$
  and distributions $\distthree, \distfive_{\termfour, \termthree}$ and
  derivations $\dertwo,\derfour_{\termfour, \termthree}$ 
  exist as in
  Lemma~\ref{lemma:structn}, we can observe,
  by induction hypothesis, that there exists a 
  distribution $\distsix$
  such that $\bsemn{\termtwo}{\distsix}$,    where
  $\distsix\geq\distthree$.
  Then we have that 
  $\supp{\distsix}\supseteq \supp{\distthree}$.
  Now, for every $\abstr{\varone}{\termfour}\in \supp{\distthree}$, suppose
  $\disteight_{\termfour, \termthree}$ is such that
  $\bsemn{\subst{\termfour}{\termthree}{\varone}}{\disteight_{\termfour, \termthree}}$.
  Again by induction hypothesis, we obtain that
  $\disteight_{\termfour, \termthree}\geq\distfive_{\termfour, \termthree}$.
  Finally:
  \begin{align*}
    \distone&\leq\sum_{\abstr{\varone}{\termfour}\in\supp{\distthree}}
      \distthree(\abstr{\varone}{\termfour})\cdot\distfive_{\termfour, \termthree}\\
    &\leq \sum_{\abstr{\varone}{\termfour}\in\supp{\distsix}}
      \distthree(\abstr{\varone}{\termfour})\cdot\distfive_{\termfour, \termthree}\\
    &\leq \sum_{\abstr{\varone}{\termfour}\in\supp{\distsix}}
      \distsix(\abstr{\varone}{\termfour})\cdot\disteight_{\termfour, \termthree}\\
      &=\disttwo.
  \end{align*}
\item
If $\termone$ is $\ps{\termtwo}{\termthree}$, we can proceed exactly as for Proposition~\ref{prop:leqbs}, sum case.
\end{varitemize}
\noindent
This concludes the proof.
\end{proof}

\begin{lemma}\label{lemma:structreversen}
Let $\termone\in\LOP$ be any term. Then: 
\begin{varenumerate}
\item\label{case1sb1n} 
  If $\termone$ is a value $\valone$ and $\ssemn{\termone}{\distone}$, then $\distone\leq\{\valone^1\}$.
\item\label{case2sb1n} 
  If $\termone$ is an application $\app{\termtwo}{\termthree}$, 
  $\ssemn{\termtwo}{\disttwo}$, 
  $\{\ssemn{\subst{\termfour}{\termthree}{\varone}}{\distfour_{\termfour}}\}_{\abstr{\varone}{\termfour}\in\supp{\disttwo}}$ 
  then there exists a distribution $\distone$ such that $\ssemn{\termone}{\distone}$ and 
  $\distone\geq\sum_{\abstr{\varone}{\termfour}\in\supp{\disttwo}}\disttwo(\abstr{\varone}{\termfour})\cdot\distfour_{\termfour}$.
\item \label{case3sb1n} 
  If $\termone$ is a sum $\ps{\termtwo}{\termthree}$, $\ssemn{\termtwo}{\disttwo}$, 
  $\ssemn{\termthree}{\distthree}$ then there exists a distribution $\distone$ such that $\ssemn{\termone}{\distone}$ and 
  $\distone\geq\frac{1}{2}\cdot\disttwo+\frac{1}{2}\cdot\distthree$.
\end{varenumerate}
\end{lemma}
\begin{proof}
Let us prove the three statements separately:
\begin{varitemize}
\item
  If $\termone$ is a value $\valone$, then the only possible judgments involving $\termone$
  are $\ssemn{\termone}{\{\valone^1\}}$ and $\ssemn{\termone}{\emdist}$. The thesis trivially holds.
\item 
  If $\termone$ is an application $\app{\termtwo}{\termthree}$, then
  we prove statement~\ref{case2sb1} by induction on the derivations 
  $\ssemn{\termtwo}{\disttwo}$ and $\ssemn{\termthree}{\distthree}$. Let's distinguish some cases:
  \begin{varitemize}
  \item  
    If $\ssemn{\termtwo}{\emdist}$, then 
    $\sum_{}\disttwo(\abstr{\varone}{\termfour})\cdot\distfour_{\termfour,\valone}=\emdist$. 
    We derive $\ssemv{\termone}{\emdist}$ by means of the rule $\sen$ and the thesis holds.
  \item 
    If $\ssemn{\termtwo}{\disttwo}\neq\emdist$, suppose that the least rule applied in the derivation is 
    $$\brule
    {
      \termtwo\rv\mul{\termsix}
    }
    {
      \ssemn{\termsix_i}{\disttwo_i}
    }
    {
      \ssemn{\termtwo}{\sum_{i=1}^n\frac{1}{n}\disttwo_i}
    }
    {\smn}
    $$
    Note that $\disttwo=\sum_{i=1}^{n}\frac{1}{n}\disttwo_{i}$.
    It is possible to apply the induction hypothesis on every ${\app{\termsix_{i}}{\termthree}}$:
    if  $\ssemn{\termsix_{i}}{\disttwo_i}$, 
    and $\ssemn{\subst{\termfour}{\termthree}{\varone}}{\distfour_{\termfour}$ 
    (for $\abstr{\varone}{\termfour}\in\supp{\disttwo_i}\subseteq\supp{\disttwo}}$) 
    follows that  $\ssemn{\app{\termsix_{i}}{\termthree}}{\distfive_i}$ 
    with $\distfive_i \geq\sum_{\abstr{\varone}{\termfour}\in\supp{\disttwo_i}}
    \disttwo_i(\abstr{\varone}{\termfour})\cdot\distfour_{\termfour}$.
    We are able to construct a derivation  $\distone$ in the following way:
    $$
    \brule
        {
          \app{\termtwo}{\termthree}\rv{\app{\mul{\termsix}}{\termthree}}
        }
        {
          \ssemn{\app{\termsix_i}\termthree}{\distfive_i}
        }
        {
          \ssemn{\app{\termtwo}{\termthree}}{\sum_{i=1}^{n}\frac{1}{n}\distfive_i}
        }
        {\smn}
        $$
    And finally we have:
   \begin{align*}
    \distone&=\sum_{i=1}^{n}\frac{1}{n}\distfive_{i}\geq \sum_{i=1}^{n}\frac{1}{n}
       \left( \sum_{\abstr{\varone}{\termfour}\in\supp{\disttwo_i}}
       \disttwo_i(\abstr{\varone}{\termfour}) \cdot\distfour_{\termfour}\right)\\
    &=\sum_{i=1}^{n}\frac{1}{n}
       \left( \sum_{\abstr{\varone}{\termfour}\in\supp{\disttwo}}
       \disttwo_i(\abstr{\varone}{\termfour}) \cdot\distfour_{\termfour}\right)\\
    &=\sum_{\abstr{\varone}{\termfour}\in\supp{\disttwo}}
       \left(\sum_{i=1}^{n}\frac{1}{n}\left(\disttwo_i(\abstr{\varone}{\termfour})\right) \cdot
       \distfour_{\termfour}\right)\\
    &=\sum_{\abstr{\varone}{\termfour}\in\supp{\disttwo}}\disttwo
       (\abstr{\varone}{\termfour}) \cdot\distfour_{\termfour}.
   \end{align*}
\item 
  Other cases can be handled similarly to the previous one.
\end{varitemize}
\item 
  $\termone$ is a sum $\ps{\termtwo}{\termthree}$. 
  The proof is trivial.
\end{varitemize}
\end{proof}

\begin{lemma}\label{lemma:sb2n}
 For each term $\termone\in\LOP$, then
\begin{varenumerate}
\item\label{case1sb2n} 
  If $\termone$ is a value $\valone$, then $\Ssemn{\termone}=\{\valone^1\}$;
\item\label{case2sb2n} 
  If $\termone$ is an application $\app{\termtwo}{\termthree}$, then 
  $\Ssemn{\termone}=\sum_{\abstr{\varone}{\termfour}\in\supp{\Ssemn{\termtwo}}} \Ssemn{\termtwo}(\abstr{\varone}{\termfour})\cdot
 \Ssemn{\subst{\termfour}{\termthree}{\varone}}$;
\item\label{case3sb2n} 
  If $\termone$ is a sum $\ps{\termtwo}{\termthree}$, then
  $\Ssemn{\termone}=\frac{1}{2}\cdot \Ssemn{\termtwo}+
  \frac{1}{2}\cdot \Ssemn{\termthree}$.
\end{varenumerate}
\end{lemma}
\begin{proof}
The inequalities above can be proved separately:
\begin{varitemize}
\item
  If $\termone$ is a value, the thesis follows by small step semantics rules. Indeed,
  $\ssemn{\termone}{\emdist}$. 
\item
  For the other cases, the 
  $\mathbf{(\leq)}$ 
  direction follows from Lemma~\ref{lemma:structn} and Fact~\ref{fact:one};
  $\geq$ direction follows from Lemma~\ref{lemma:structreversen}  and Fact~\ref{fact:one}.
\end{varitemize}
\end{proof}

\begin{proposition}\label{prop:supbsemn}
$\bsemn{\termone}{\Ssemn{\termone}}$.
\end{proposition}
\begin{proof} 
We will prove the thesis by coinduction:
We can prove that all judgments $\bsemn{\termone}{\Ssemv{\termone}}$
belong to the coinductive interpretation of the underlying formal system $\isone$ 
(in this case, the formal system is $\isone=\{{\bvn}, {\ban}, {\bsn}\}$).
To do that, we need to prove that the set $\sjone$ of all those judgment
is consistent, i.e. that $\sjone\subseteq\opinf{\isone}(\sjone)$.
This amounts to show that if $\juone\in\sjone$, then there is a derivation
for $\juone$ whose immediate premises are themselves in $\sjone$. Let's
distinguish some cases:
\begin{varitemize}
\item 
  If $\termone=\valone$ then $\bsemn{\valone}{\{V^{1}\}}$ by $\mathsf{bv_v}$ rule, 
  and $\{V^{1}\}= \Ssemn{\valone}$ because of Lemma~\ref{lemma:sb2n}.
\item 
  If $\termone$ is an application $\app{\termtwo}{\termthree}$, take the 
  judgment $\juone_1=\bsemn{\termtwo}{\Ssemn{\termtwo}}$, 
  and the family of judgments 
  $\{\bsemn{\{\subst{\termfour}{\termthree}{\varone}\}}
  {\Ssemn{\{\subst{\termfour}{\termthree}{\varone}\}}}\}_{\abstr{\varone}{\termfour}\in\supp{\Ssemn{\termtwo}}}$: 
  we will prove that the judgment $\bsemn{\app{\termtwo}{\termthree}}{\Ssemn{\app{\termtwo}{\termthree}}}$ can be derived in a 
  single step from $\juone_1$ and those in the family above
  by means of $\ban$ rule.
  Simply observe that
  $$
  \brule
      {\bsemn{\termtwo}{\Ssemn{\termtwo}}}
      {\{\bsemn{\{\subst{\termfour}{\termthree}{\varone}\}}
  {\Ssemv{\{\subst{\termfour}{\termthree}{\varone}\}}}\}_{\abstr{\varone}{\termfour}\in\supp{\Ssemv{\termtwo}}}}
      {\bsemn{\app{\termtwo}{\termthree}}{\sum_{\abstr{\varone}{\termfour}\in\supp{\Ssemn{\termtwo}}} \Ssemn{\termtwo}(\abstr{\varone}{\termfour})\cdot \Ssemn{\subst{\termfour}{\termthree}{\varone}}}}
      {{\bav}}
  $$ 
  The thesis follows applying Lemma~\ref{lemma:sb2n}.
\item 
  If $\termone$ is a sum $\ps{\termtwo}{\termthree}$, take the judgment $c_1= \bsemn{\termtwo}{\Ssemn{\termtwo}}$ 
  and $\juone_2= \bsemn{\termone}{\Ssemn{\termthree}}$: we will prove that the judgment 
  $\bsemn{\ps{\termtwo}{\termthree}}{\Ssemn{\app{\termtwo}{\termthree}}}$ can be inferred in a single step from 
  $\juone_1$ and $\juone_2$ by means of $\mathsf{bs_v}$ rules.
  Clearly, $\juone_1$ and $\juone_2$ belong to $\sjone$. Moreover
  $$
  \brule
      {\bsemn{\termtwo}{\Ssemn{\termtwo}}}
      {\bsemn{\termthree}{\Ssemn{\termthree}}}
      {\bsemn{\ps{\termtwo}{\termthree}}{\frac{1}{2}\cdot\Ssemv{\termtwo}}+\frac{1}{2}\cdot \Ssemn{\termthree}}
      {{\bsn}}
  $$ 
  and by  Lemma~\ref{lemma:sb2n}, case~\ref{case3sb2n} we obtain the thesis. 
\end{varitemize}
This concludes the proof.
\end{proof}

\begin{theorem}
$\Bsemn{\termone}=\Ssemn{\termone}$.
\end{theorem}
\begin{proof}
This is a corollary of Proposition~\ref{prop:supbsemn} and Proposition~\ref{prop:leqbsn}.
\end{proof}

\section{CPS Translations and Simulations}\label{sec:trans}

In this section we show that in $\LOP$  it is possible to simulate call-by-value by call-by-name and vice versa. 
To do that, we follow Plotkin's CPS translation~\cite{Plotkin75}, extended to accommodate binary choice.

It is well known that in the \emph{weak} untyped lambda calculus, call-by-value and 
call-by-name are not equivalent notions of reduction. 
Moreover, the presence of binary choice exacerbates the confluence problem, as shown in Section~\ref{sec:obs}.
It is clear that duplications play a central role: in particular, the order in which 
probabilistic choices and duplications are performed matters.

The presence of binary choices opens a related question about the possibility of including some  
``administrative rules'' distributing sums over the other constructs. 
The following example shows that it may be critical to include administrative rules:
\begin{example}\label{ex:adm}
Let us consider the following terms: $\termone\equiv \idt{(\ps{\termtt}{\termff}})$ and  $\termtwo\equiv 
\ps{(\idt{\termtt})}{(\idt{\termff})}$. Note that $\termone$ can 
be obtained from $\termtwo$ by means of administrative rules like 
$\termthree(\ps{\termfour}{\termfive})= \ps{(\termthree\termfour)}{(\termthree\termfive)}$. 
Consider the term  
$\termxorcont\equiv(\abstr{\varthree}{\termxor\varthree\varthree})$.
We have 
$\Ssemv{\termone\termxorcont}=\{\termff^1\}=\distone$ and 
$\Ssemn{\termone\termxorcont}=\{\termff^{\frac{1}{2}},\termtt^{\frac{1}{2}}\}=
\disttwo$, whereas $\Ssemv{\termtwo\termxorcont}=\Ssemn{\termtwo\termxorcont}=\distone$ and therefore
$\termone $ and $\termtwo $ have different call-by-name observational behaviors, 
thus $\termone$ and $\termtwo$ can not be considered equivalent terms.
\end{example}
A study about the observational behavior of terms is a fascinating subject, but is out of the scope of this paper.
By the way, this is extensively investigated in the non deterministic setting~\cite{deLigPip95}, in which an 
algebraic semantics of terms is defined by way of a generalization of B\"ohm's trees.
A further generalization to the probabilistic setting is left to future work.

What we are interested here is to develop an operational study of $\LOP$. And an interesting
question is clearly whether call-by-value and call-by-name, although being distinct notions of reduction,
can be somehow made equivalent through a suitable CPS translation, even in presence of
binary choices (with a probabilistic semantics).

In this section, a simulation between call-by-value and call-by-name in $\LOP$ is proved.
We begin with the simulation of call-by-value by call-by-name (Section~\ref{sec:cbvbycbn}), 
then we will carry on with the simulation of call-by-name by call-by-value (Section~\ref{sec:cbnbycbv}).
\subsection{Simulating Call-by-Value with Call-By-Name}\label{sec:cbvbycbn}
Suppose to extend $\LOP$ with a denumerable set of continuation variables $\setvarg=\{\svarone,\svartwo,\svarfive\ldots\}$, 
disjoint from the original set $\setvar$ of variables of the language. We will call $\LPplus$ 
the language of lambda terms extended this way. All definitions and constructions on $\LOP$
(including its operational semantics) extend smoothly to $\LPplus$

Call-by-value reduction on $\LOP$ can be simulated by call-by-name reduction of $\LPplus$ by translating
every term $\termone$ in $\LOP$ to a term in $\LPplus$, which will be proved to
be equivalent to $\termone$ in a certain sense.
\begin{definition}[Call-by-value Translation]\label{def:mapvn}\mbox{}
The translation map $ \mapvn{\;\cdot\;}$ from $\LOP$ to $\LPplus$ is recursively defined as follows:
\begin{align*}
\mapvn{\varone}&=\lambda \svarfive.\app{\svarfive}{\varone}; &
\mapvn{\termone\termtwo}&=\lambda\svarfive.(\mapvn{\termone}(\lambda\svarone.\mapvn{\termtwo}(\lambda\svartwo.\svarone\svartwo\svarfive)));\\
\mapvn{\abstr{\varone}{\termone}}&=\lambda\svarfive.\app{\svarfive}{\lambda x.\mapvn{\termone}};&
\mapvn{\ps{\termone}{\termtwo}}&=\lambda\svarfive. 
   \mapvn{\termone}(\abstr{\svarone}{\mapvn{\termtwo}}(\abstr{\svartwo}{\app{(\ps{(\abstr{\svarthree}{\app{\svarthree}{\svarone}})}
   {(\abstr{\svarthree}{\app{\svarthree}{\svartwo}})})}{\svarfive}})).
\end{align*}
\end{definition}
As expected, call-by-value is simulated by way of so-called continuations.
As a consequence, (call-by-value) reduction on $\termone$ is \emph{not} simulated
simply reducing (in call-by-name), $\mapvn{\termone}$, but by feeding it with
the identity continuation $\abstr{\varone}{\varone}$. Furthermore, we do not 
obtain this way the same value(s) as the one(s) we would obtain by evaluating
$\termone$, but something related to that by a function $\Psi$, 
which sends values in $\LOP$ into values in $\LPplus$ as follows: 
\begin{align*}
\Psi(x)&=x; &\Psi(\abstr{\varone}{\termone})&=\abstr{\varone}{\mapvn{\termone}}.
\end{align*}
Clearly, $\Psi$ can be naturally extended to a map on distributions (of values). 

The rest of this section is devoted to showing
the following theorem, which retraces in $\LOP$ the classic result from~\cite{Plotkin75}:  
\begin{theorem}[Simulation]\label{th:simVbyN}
For every $\termone $, $\Psi(\Ssemv{\termone})=\Ssemn{\mapvn{\termone}(\idt)}$.
\end{theorem}

Let us get back to the proof. The following fact and the following substitution lemma are necessary.
\begin{fact}
$(\ps{\termone}{\termtwo})\{\valone/\varone\}= \ps{\termone\{\valone/\varone\}}{\termtwo\{\valone/\varone\}}$.
\end{fact}

Lemma~\ref{lemma:subst} extends Plotkin's substitution lemma (see ~\cite{Plotkin75}, Lemma 1, page 149)
\begin{lemma}[Substitution]\label{lemma:subst}
$\mapvn{\termone}\{\Psi(\valone)/\varone\}= \mapvn{\termone\{\valone/\varone\}}$.
\end{lemma}
\begin{proof}
By induction on the size of $\termone$. An interesting case:
\begin{varitemize}
\item 
  If $\termone$ is $\ps{\termtwo}{\termthree}$, then 
  \begin{align*}
    \mapvn{\termone}\{\Psi(\valone)/\varone\}&=\lambda\epsilon. \mapvn{\termtwo}(\lambda\alpha. 
          \mapvn{\termthree}(\abstr{\svartwo}{(\app{(\ps{(\abstr{\svarthree}{\app{\svarthree}{\svarone}})}
          {(\abstr{\svarthree}{\app{\svarthree}{\svartwo}})})}{\svarfive})})\{\Psi(\valone)/\varone\}\\ 
       &=\lambda\epsilon.\mapvn{\termtwo}\{\Psi(\valone)/\varone\}(\lambda\alpha. \mapvn{\termthree}\{\Psi(\valone)/\varone\}
          (\abstr{\svartwo}{(\app{(\ps{(\abstr{\svarthree}{\app{\svarthree}{\svarone}})}
          {(\abstr{\svarthree}{\app{\svarthree}{\svartwo}})})}{\svarfive})})\\
       &=\lambda\epsilon.\mapvn{\termtwo\{\Psi(\valone)/\varone\}}(\lambda\alpha. \mapvn{\termthree\{\Psi(\valone)/\varone\}}
          (\abstr{\svartwo}{(\app{(\ps{(\abstr{\svarthree}{\app{\svarthree}{\svarone}})}
          {(\abstr{\svarthree}{\app{\svarthree}{\svartwo}})})}{\svarfive})})\\
       &=\mapvn{\ps{\termtwo\{\Psi(\valone)/\varone\}}{\termthree\{\Psi(\valone)/\varone\}}}=\mapvn{\ps{\termtwo}{\termthree}\{\Psi(\valone)/\varone\}}.
  \end{align*}
\end{varitemize}
\item Other cases are very similar to the previous one.
This concludes the proof.
\end{proof}
We define now the suitable extension of the infix operator ``$:$'' introduced by Plotkin in \cite{Plotkin75}.

\begin{definition}[Infix operator ``:" for $\mapvn{\cdot}$]\label{def:infix1}\mbox{}
The infix operator ``:'' for the map $\mapvn{\cdot}$ is defined as follows:
\begin{align*}
  \valone:\ctone&=\ctone\Psi(\valone)\\
  \app{\termthree}{\termfour}:\ctone&=\termthree:(\lambda\alpha. \mapvn{\termfour}(\lambda\beta.\alpha\beta \ctone))\mbox{ if }\termthree\notin\val\\
  \app{\valone}{\termthree}:\ctone&=\termthree:((\lambda\beta.\Psi(\valone)\beta \ctone))\mbox{ if }\termthree\notin\val\\
  \app{\valone}{\valtwo}:\ctone&=\Psi(\valone)\Psi(\valtwo) \ctone\\
  \termthree\oplus_{}\termfour : \ctone&= \termthree:(\lambda\alpha. \mapvn{\termfour}
    (\abstr{\svartwo}{(\app{(\ps{(\abstr{\svarthree}{\app{\svarthree}{\svarone}})}
    {(\abstr{\svarthree}{\app{\svarthree}{\svartwo}})})}{\ctone})}))
    \mbox{ if } \termthree\notin\val\\
  \valone\oplus_{}\termthree : \ctone&= \termthree:
    (\abstr{\svartwo}{(\app{(\ps{(\abstr{\svarthree}{\app{\svarthree}{\Psi(\valone)}})}
    {(\abstr{\svarthree}{\app{\svarthree}{\svartwo}})})}{\ctone})})
    \mbox{ if } \termthree\notin\val\\
  \valone\oplus_{}\valtwo : \ctone&=\app{(\ps{(\abstr{\svarthree}{\app{\svarthree}{\Psi(\valone)}})}
    {(\abstr{\svarthree}{\app{\svarthree}{\Psi(\valtwo)}})})}{\ctone}
\end{align*}
\end{definition}
The operator ``:'' can be naturally extended to distributions of values. In this case we will use the notation $\distone: \ctone$.

The following lemmas exploit operator ``$:$'' and give an important intermediate results in order to prove Theorem~\ref{th:simVbyN}. 
As previously declared, Lemma~\ref{lemma:redK1} shows that $\adm{\termone}{\ctone}$ is the result of the computation from 
$\mapvn{\termone}\ctone$ involving all structural reductions:
\begin{lemma}\label{lemma:redK1}
For all $\termone\in\LOP$, $\mapvn{\termone}\ctone\rn^{*}\termone:\ctone$.
\end{lemma}
\begin{proof}
By induction on $\termone$ and on the definition of ``$:$''. Some interesting cases:
\begin{varitemize}
\item 
  If $\termone$ is a value, then 
  $$
  \mapvn{\termone}\ctone=\app{(\lambda\epsilon.\epsilon\Psi(\termone))}{\ctone}\rn\ctone\Psi(\termone)
   =\termone:\ctone.
  $$
\item
  If $\termone$ is $\ps{\termthree}{\termsix}$. There are three subcases:
  \begin{varenumerate}
     \item 
       $\termthree\notin\val$. Then
         \begin{align*}
           \mapvn{\ps{\termthree}{\termsix}}\ctone&=(\lambda\epsilon.\mapvn{\termthree}(\lambda\alpha. 
              \mapvn{\termsix}(\abstr{\svartwo}{(\app{(\ps{(\abstr{\svarthree}{\app{\svarthree}{\svarone}})}
              {(\abstr{\svarthree}{\app{\svarthree}{\svartwo}})})}{\svarfive})})))\ctone\\
              &\rn\mapvn{\termthree}(\lambda\alpha. 
              \mapvn{\termsix}(\abstr{\svartwo}{(\app{(\ps{(\abstr{\svarthree}{\app{\svarthree}{\svarone}})}
              {(\abstr{\svarthree}{\app{\svarthree}{\svartwo}})})}{\ctone})}))\\
              &\rn^*\termthree:(\lambda\alpha. 
              \mapvn{\termsix}(\abstr{\svartwo}{(\app{(\ps{(\abstr{\svarthree}{\app{\svarthree}{\svarone}})}
              {(\abstr{\svarthree}{\app{\svarthree}{\svartwo}})})}{\ctone})}))\\
              &= \termone:\ctone
         \end{align*}
     \item 
       $\termthree\in\val$ and $\termsix\notin\val$. Then:
       \begin{align*}
         \mapvn{\ps{\termthree}{\termsix}}\ctone&=(\lambda\svarfive.\mapvn{\termthree}(\lambda\alpha. 
            \mapvn{\termsix}(\abstr{\svartwo}{(\app{(\ps{(\abstr{\svarthree}{\app{\svarthree}{\svarone}})}
            {(\abstr{\svarthree}{\app{\svarthree}{\svartwo}})})}{\svarfive})})))\ctone\\
         &\rn \mapvn{\termthree}(\lambda\alpha. 
            \mapvn{\termsix}(\abstr{\svartwo}{(\app{(\ps{(\abstr{\svarthree}{\app{\svarthree}{\svarone}})}
            {(\abstr{\svarthree}{\app{\svarthree}{\svartwo}})})}{\ctone})}))\\
         &\rn^* \termthree:(\lambda\alpha. 
            \mapvn{\termsix}(\abstr{\svartwo}{(\app{(\ps{(\abstr{\svarthree}{\app{\svarthree}{\svarone}})}
            {(\abstr{\svarthree}{\app{\svarthree}{\svartwo}})})}{\ctone})}))\\
         &= (\lambda\svarone. 
            \mapvn{\termsix}(\abstr{\svartwo}{(\app{(\ps{(\abstr{\svarthree}{\app{\svarthree}{\svarone}})}
            {(\abstr{\svarthree}{\app{\svarthree}{\svartwo}})})}{\ctone})}))\Psi(\termthree)\\
         &\rn \mapvn{\termsix}(\abstr{\svartwo}{(\app{(\ps{(\abstr{\svarthree}{\app{\svarthree}{\Psi(\termthree)}})}
            {(\abstr{\svarthree}{\app{\svarthree}{\svartwo}})})}{\ctone})})\\
         &\rn^* \termsix:(\abstr{\svartwo}{(\app{(\ps{(\abstr{\svarthree}{\app{\svarthree}{\Psi(\termthree)}})}
            {(\abstr{\svarthree}{\app{\svarthree}{\svartwo}})})}{\ctone})})= \termone:\ctone.
       \end{align*}
     \item 
       $\termthree, \termsix\in\val$. Then:
       \begin{align*}
         (\mapvn{\ps{\termthree}{\termsix}})\ctone&=\lambda\epsilon.\mapvn{\termthree}(\lambda\alpha. \mapvn{\termsix}(\lambda\beta.(\ps{\epsilon\alpha}{\epsilon\beta})))\ctone\\
         &\rn \mapvn{\termthree}(\lambda\alpha. \mapvn{\termsix}(\lambda\beta.(\ps{\ctone\alpha}{\ctone\beta})))\\
         &\rn{\termthree}:(\lambda\alpha. \mapvn{\termsix}(\lambda\beta.(\ps{\ctone\alpha}{\ctone\beta})))\\
         &=(\lambda\alpha. \mapvn{\termsix}(\lambda\beta.(\ps{\ctone\alpha}{\ctone\beta})))(\Psi(\termthree))\\
         &\rn \mapvn{\termsix}(\lambda\beta.(\ps{\ctone\Psi(\termthree)}{\ctone\beta})))\\
         &\rn\termsix:(\lambda\beta.(\ps{\ctone\Psi(\termthree)}{\ctone\beta})))=\\
         & (\lambda\beta.(\ps{\ctone\Psi(\termthree)}{\ctone\beta})))(\Psi(\termsix))\\
         &\rn(\ps{\ctone\Psi(\termthree)}{\ctone\Psi(\termsix)})=\ps{\termthree}{\termsix}:\ctone.
        \end{align*}
    \end{varenumerate}
\end{varitemize}
This concludes the proof.
\end{proof}

As a technical tool, we here need a generalization of the small step semantic relation $\ssemn{\cdot}{\cdot}$,
that we denote as $\vssemn{\cdot}{\cdot}$. The new relation is defined inductively as $\ssemn{\cdot}{\cdot}$, but  
the first two rules are replaced by the two rules
$$
\zrule
    {
      \vssemn{\termone}{\{\termone^1\}}
    }
    {}
$$
$$
\zrule
    {
      \vssemn{\termone}{\emdist}
    }
    {}
$$
This means, in particular, that the relation above maps terms to distribution over terms, rather than
distributions over values.

Of course, relations $\ssemn{\cdot}{\cdot}$ and  $\vssemn{\cdot}{\cdot}$ are strongly related, as shown in the following lemma:
\begin{lemma}\label{lemma:varvsnorm}
If $\vssemn{\termone}{\distone}$, where $\supp{\distone}\subseteq\val$, then
$\ssemn{\termone}{\distone}$. Conversely, if $\ssemn{\termone}{\distone}$
then $\vssemn{\termone}{\distone}$.
\end{lemma}
\begin{proof}
Simple inductions.
\end{proof}

Lemma~\ref{lemma:auxredK} is an auxiliary tool for the proof of Lemma~\ref{lemma:redK2}. In particular, these 
lemmas extend Plotkin's lemma which states that whenever a term $\termone$ reduces in call by value to a term 
$\termtwo$, then there exists a call-by-name computation from $\termone:\ctone$ to $\termtwo:\ctone$. In $\LOP$
 we have to deal with the non deterministic nature of our relations $\rv$ and $\rn$ (which map terms into sequence 
of terms), and with the fact that the evaluation of a term returns a probability distribution on values and not 
a single values.

\begin{lemma}\label{lemma:auxredK}
If $\termone\rv\termtwo_1,\ldots,\termtwo_\natone$ and 
$\ctone$ is a value, then there are terms $\termfour_1,\ldots,\termfour_\natone$ such that
$\adm{\termone}{\ctone}\rn^* \termfour_1,\ldots,\termfour_\natone $
and $\termfour_i\rn^*\adm{\termtwo_i}{\ctone}$ for every $i\in\{1,\ldots,\natone\}$.
\end{lemma}
\begin{proof}
By induction on the structure of $\termone$. Some interesting cases.
\begin{varitemize}
\item 
  If $\termone=\ps{\termthree}{\termsix}$ with $\termthree, \termsix\in\val$, 
  then $\termone\rv\termthree, \termsix $ and by definition
  \begin{align*}
    \adm{\termone}{\ctone}&=\app{(\ps{(\abstr{\svarthree}{\app{\svarthree}{\Psi(\termthree)}})}
    {(\abstr{\svarthree}{\app{\svarthree}{\Psi(\termsix)}})})}{\ctone}\\
    &\rn \app{(\abstr{\svarthree}{\app{\svarthree}{\Psi(\termthree)}})}{\ctone},
         \app{(\abstr{\svarthree}{\app{\svarthree}{\Psi(\termsix)}})}{\ctone}\\
    \app{(\abstr{\svarthree}{\app{\svarthree}{\Psi(\termthree)}})}{\ctone}&\rn \ctone\Psi(\termthree)=\termthree:\ctone\\
    \app{(\abstr{\svarthree}{\app{\svarthree}{\Psi(\termsix)}})}{\ctone}&\rn \ctone\Psi(\termsix)=\termsix:\ctone
  \end{align*}
\item 
  If  $\termone=\ps{\termthree}{\termsix}$ with $\termthree\notin\val$, 
  then $\termthree\rv\termseven_1,\ldots,\termseven_\natone$ and
  $\termtwo_i=\ps{\termseven_i}{\termsix}$.
  But by induction hypothesis, there are terms $\termfive_1,\ldots,\termfive_\natone$ such that
  \begin{align*}
    \adm{\termone}{\ctone}&=\adm{\termthree}{(\lambda\alpha. \mapvn{\termsix}
      (\abstr{\svartwo}{(\app{(\ps{(\abstr{\svarthree}{\app{\svarthree}{\svarone}})}
      {(\abstr{\svarthree}{\app{\svarthree}{\svartwo}})})}{\ctone})}))}=\adm{\termthree}{\cttwo}\\
      &\rn\termfive_1,\ldots,\termfive_\natone\\
    \termfive_i&\rn^*\adm{\termseven_i}{\cttwo}
  \end{align*}
  Notice that, on the other hand:
  $$
    \adm{\termseven_i}{\cttwo}=\adm{\termseven_i}{(\lambda\alpha. \mapvn{\termsix}
      (\abstr{\svartwo}{(\app{(\ps{(\abstr{\svarthree}{\app{\svarthree}{\svarone}})}
      {(\abstr{\svarthree}{\app{\svarthree}{\svartwo}})})}{\ctone})}))}\\
  $$
  Now, if $\termseven_i$ is a value, then
  \begin{align*}
    \adm{\termseven_i}{\cttwo}&=\app{\cttwo}{\Psi(\termseven_i)}\rn\mapvn{\termsix}
      (\abstr{\svartwo}{(\app{(\ps{(\abstr{\svarthree}{\app{\svarthree}{\Psi(\termseven_i)}})}
      {(\abstr{\svarthree}{\app{\svarthree}{\svartwo}})})}{\ctone})})\\
    &=\adm{\ps{\termseven_i}{\termsix}}{\ctone}
  \end{align*}
  If $\termseven_i$ is not a value, then $\adm{\termseven_i}{\cttwo}$ is itself
  $\adm{\ps{\termseven_i}{\termsix}}{\ctone}$.
\item 
  If  $\termone=\ps{\termthree}{\termsix}$ with $\termthree\in\val$ then 
  $\termsix\rv\termseven_1,\ldots,\termseven_\natone$ and
  $\termtwo_i=\ps{\termthree}{\termseven_i}$.
Similar to the previous cases.
\end{varitemize}
This concludes the proof.
\end{proof}

\begin{lemma}\label{lemma:redK2}
If $\ssemv{\termone}{\distone}$ and $\ctone$ closed value, then
$\vssemn{\adm{\termone}{\ctone}}{\adm{\distone}{\ctone}}$.
\end{lemma}
\begin{proof}
By induction on the complexity of a derivation for $\ssemv{\termone}{\distone}$:
\begin{varitemize}
\item
  If $\ssemv{\termone}{\emdist}$, then clearly
  $\vssemn{\adm{\termone}{\ctone}}{\adm{\emdist}{\ctone}}$,
  simply because $\adm{\emdist}{\ctone}=\emdist$.
\item
  If $\termone\rv\mul{\termtwo}$ and $\ssemv{\termtwo_i}{\disttwo_i}$,
  where $\distone=\sum_{i=1}^\natone\frac{1}{\natone}\disttwo_i$, then,
  by Lemma~\ref{lemma:auxredK},
  $\adm{\termone}{\ctone}\rn\termthree_1,\ldots,\termthree_\natone$ and
  $\termthree_i\rn^*\adm{\termtwo_i}{\ctone}$.
  By induction hypothesis,
  $\vssemn{\adm{\termtwo_i}{\ctone}}{\adm{\disttwo_i}{\ctone}}$ and
  so we can easily form a derivation for
  $\vssemn{\adm{\termone}{\ctone}}{\sum_{i=1}^\natone(\frac{1}{\natone}\adm{\disttwo_i}{\ctone})}$
  But clearly, 
  $$
  \sum_{i=1}^\natone\left(\frac{1}{\natone}\adm{\disttwo_i}{\ctone}\right)=
  \adm{\left(\sum_{i=1}^\natone\frac{1}{\natone}\disttwo_i\right)}{\ctone}
  $$
\end{varitemize}
This concludes the proof.
\end{proof}

Lemma~\ref{lemma:redK3} shows that  the call-by-value evaluation of a term $\termone$  dominates, in term of distributions, 
the call by-name evaluation  of the continuation  $\adm{{\termone}}{\idt}$, up to a suitable translation by function $\Psi$:

\begin{lemma}\label{lemma:redK3}
If $\ssemn{\adm{\termone}{\idt}}{\distone}$, then $\ssemv{\termone}{\disttwo}$, where
$\Psi(\disttwo)\geq\distone$.
\end{lemma}
\begin{proof}
The proof goes by induction on the complexity of a derivation $\derone$ of
$\ssemn{\adm{\termone}{\abstr{\varone}{\varone}}}{\distone}$. Some interesting cases:
\begin{varitemize}
\item If $\termone$ is a value $\valone$:
\begin{varitemize}
\item
  if $\ssemn{\adm{\valone}{\ide}}{\emdist}$, then $\disttwo=\emdist$ itself;
  \item
  if $\ssemn{\adm{\valone}{\ide}}{\distone}$ with $\distone\neq\emdist$,  then $\distone=\{\valone^{1}\}$.
\end{varitemize}
\item If $\termone$ is not a value:
\begin{varitemize}
\item if $\ssemn{\adm{\termone}{\ide}}{\emdist}$ the thesis follows trivially;
\item if $\ssemn{\adm{\termone}{\ide}}{\distone}$ with $\distone\neq\emdist$, let
  $\termtwo_1,\ldots\termtwo_\natone$ be such that
  $\termone\rv\mul{\termtwo}$. By Lemma~\ref{lemma:auxredK},
  there must be derivations $\dertwo_1,\ldots,\dertwo_n$ (all smaller than $\derone$)
  such that $\dertwo_i:\ssemn{\adm{\termtwo_i}{\idt}}{\disttwo_i}$ and
  $\distone=\sum_{i=1}^\natone\frac{1}{\natone}\disttwo_i$. By induction hypothesis,
  $\ssemv{\termtwo_i}{\distthree_i}$ where $\Psi(\distthree_i)\geq\disttwo_i$ and
  so we can form a derivation of $\ssemv{\termone}{\left(\sum_{i=1}^\natone\distthree_i\right)}$.
  But clearly,
  $$
  \Psi(\sum_{i=1}^\natone\distthree_i)=\sum_{i=1}^\natone\Psi(\distthree_i)\geq\sum_{i=1}^\natone\disttwo_i=\distone.
  $$
\end{varitemize}
\end{varitemize}
This concludes the proof.
\end{proof}

\noindent

By means of Lemma~\ref{lemma:redK1}, Lemma~\ref{lemma:redK2}, Lemma~\ref{lemma:varvsnorm} and 
Lemma~\ref{lemma:redK3}, it is possible to prove Theorem~\ref{th:simVbyN}.

\begin{proof}[Theorem~\ref{th:simVbyN}]
To prove that $\Psi(\Ssemv{\termone})=\Ssemn{\app{\mapvn{\termone}}{\idt}}$,
we proceed by showing the following:
\begin{varenumerate}
\item\label{point:vnsim1}
  If $\ssemv{\termone}{\distone}$, then $\ssemn{\app{\mapvn{\termone}}{\idt}}{\disttwo}$
  and $\disttwo\geq\Psi(\distone)$.
\item\label{point:vnsim2}
  If $\ssemn{\app{\mapvn{\termone}}{\idt}}{\distone}$, then $\ssemv{\termone}{\disttwo}$
  where $\Psi(\disttwo)\geq\distone$.
\end{varenumerate}
To prove point~\ref{point:vnsim1}, we observe that
if $\ssemv{\termone}{\distone}$, then both
$\mapvn{\termone}\ctone\rn^{*}\termone:\ctone$ (by Lemma~\ref{lemma:redK1}) and
$\ssemv{\adm{\termone}{\ctone}}{\adm{\distone}{\ctone}}$ (by Lemma~\ref{lemma:redK2} and
Lemma~\ref{lemma:varvsnorm}). Now, notice that, for every value $\valone$,
$$
\adm{\valone}{(\idt)}=\app{(\idt)}{\Psi(\valone)}\rv\Psi(\valone).
$$
As a consequence, it's clear that
$$
\ssemn{\adm{\mapvn{\termone}}{\idt}}{\distone}.
$$
Point~\ref{point:vnsim2} is nothing more than Lemma~\ref{lemma:redK3}.
This concludes the proof.
\end{proof}

\subsection{Simulating Call-By-Name by Call-By-Value}\label{sec:cbnbycbv}

In Section~\ref{sec:cbvbycbn} we proved the ``imperfect'' simulation of call-by-value by call-by name strategy. The inverse direction
is provable in a very similar way.

Let us define the extended language $\LOPplus$ as in Section~\ref{sec:cbvbycbn}.

\begin{definition}[Call-by-value Translation]\label{def:mapnv}\mbox{}
The translation map $ \mapnv{\;\cdot\;}$ from $\LOP$ to $\LPplus$ is recursively defined as follows:
\begin{align*}
\mapnv{\varone}&=\varone\\
\mapnv{\abstr{\varone}{\termone}}&=\lambda\svarfive.\app{\svarfive}{\lambda x.\mapnv{\termone}}\\
\mapnv{\termone\termtwo}&=\lambda\svarfive.\mapnv{\termone}(\lambda\svarone.\app{\app{\svarone}{\mapnv{\termtwo}}}{\svarfive})\\
\mapnv{\ps{\termone}{\termtwo}}&=
   \abstr{\epsilon}{(\app{(\ps{(\abstr{\alpha}{\app{\mapnv{\termone}}{\alpha}})}
                              {(\abstr{\alpha}{\app{\mapnv{\termtwo}}{\alpha}})})}{\epsilon})}
\end{align*}
\end{definition}
We also define a function $\Phi$, which sends values into terms as 
\begin{align*}
\Phi(\varone)&=\app{\varone}{(\abstr{\vartwo}{\vartwo})}\\
\Phi(\abstr{\varone}{\termone})&=\abstr{\varone}{\mapnv{\termone}}
\end{align*}
Observe that $\Phi$ sends closed values to closed values.
As such, it can be naturally extended to distributions of closed values. 

It is possible to state and prove the dual of Theorem~\ref{th:simVbyN}:
\begin{theorem}[Simulation]\label{th:simNbyV}
For every $\termone $, $\Phi(\Ssemn{\termone})=\Ssemv{\app{\mapnv{\termone}}{(\idt)}}$.
\end{theorem}
To prove Theorem ~\ref{th:simNbyV} we need some technical lemmas, as in Section~\ref{sec:cbvbycbn}. We retrace the same proof techniques.
\begin{lemma}\label{lemma:substitutionN}
$\mapnv{\termone}\{\mapnv{\termtwo}/\varone\}= \mapnv{\termone\{\termtwo/\varone\}}$.
\end{lemma}
\begin{proof}
By induction on the size of $\termone$. 
We give here only the sum case. See~\cite{} for other cases.
\begin{varitemize}
\item If $\termone=\ps{\termthree}{\termsix}$ then:

\begin{align*}
\mapnv{\ps{\termthree}{\termsix}}\{\mapnv{\termtwo}/\varone\}&\df \abstr{\epsilon}{(\app{(\ps{(\abstr{\alpha}{\app{\mapnv{\termthree}}{\alpha}})}
                              {(\abstr{\alpha}{\app{\mapnv{\termsix}}{\alpha}})})}{\epsilon})}\{\mapnv{\termtwo}/\varone\}\\
                              &= \abstr{\epsilon}{(\app{(\ps{(\abstr{\alpha}{\app{\mapnv{\termthree}\{\mapnv{\termtwo}/\varone\}}{\alpha}})}
                              {(\abstr{\alpha}{\app{\mapnv{\termsix}\{\mapnv{\termtwo}/\varone\}}{\alpha}})})}{\epsilon})}
                              \\
                              &\stackrel{i.h.}{=} \abstr{\epsilon}{(\app{(\ps{(\abstr{\alpha}{\app{\mapnv{\termthree\{\mapnv{\termtwo}/\varone\}}}{\alpha}})}
                              {(\abstr{\alpha}{\app{\mapnv{\termsix}\{\mapnv{\termtwo/\varone\}}}{\alpha}})})}{\epsilon})}\\
                              &\df \mapnv{(\ps{\termthree}{\termsix})\{\mapnv{\termtwo}/\varone\}}
\end{align*}
\end{varitemize}
This concludes the proof.
\end{proof}

Now the infix operator is defined as:

\begin{align*}
  \valone:\ctone&=\ctone\Phi(\valone)\\
  \app{\termthree}{\termfour}:\ctone&=\termthree:(\lambda\alpha.\alpha\mapnv{\termfour}\ctone)\mbox{ if }\termthree\notin\val\\
  \app{\valone}{\termthree}:\ctone&=\Phi(\valone)\mapnv{\termthree}\ctone\mbox{ if }\valone\in\val\mbox{ and } \termthree\notin\val\\
  \termthree\oplus_{}\termfour : \ctone&= \app{(\ps{(\abstr{\alpha}{\app{\mapnv{\termthree}}{\alpha}})}
    {(\abstr{\alpha}{\app{\mapnv{\termfour}}{\alpha}})})}{\ctone}
\end{align*}
\begin{lemma}\label{lemma:redKn1}
For all $\termone\in\LOP$, $\mapnv{\termone}\ctone\rv^{*}\termone:\ctone$.
\end{lemma}
\begin{proof}
By induction on the term $\termone$. We give here only the sum case. See~\cite{} for other cases.
\begin{align*}
\app{\mapnv{\ps{\termthree}{\termfour}}}{\ctone}&=\app{\abstr{\epsilon}{(\app{(\ps{(\abstr{\alpha}{\app{\mapnv{\termthree}}{\alpha}})}
                              {(\abstr{\alpha}{\app{\mapnv{\termfour}}{\alpha}})})}{\epsilon})}}{\ctone}\\
&\rv {}{(\app{(\ps{(\abstr{\alpha}{\app{\mapnv{\termthree}}{\alpha}})}
                              {(\abstr{\alpha}{\app{\mapnv{\termfour}}{\alpha}})})}{\ctone})}\\
                              &\df \adm{\ps{\termthree}{\termfour}}{\ctone}
\end{align*}
\end{proof}

\begin{lemma}\label{lemma:auxrednK}
If $\termone\rn\termtwo_1,\ldots,\termtwo_\natone$ and 
$\ctone$ is a value, then there are terms $\termfour_1,\ldots,\termfour_\natone$ such that
$\adm{\termone}{\ctone}\rv \termfour_1,\ldots,\termfour_\natone $
and $\termfour_i\rv^*\adm{\termtwo_i}{\ctone}$ for every $i\in\{1,\ldots,\natone\}$.
\end{lemma}
As a technical tool, we here need a generalization of the small step semantic relation $\ssemn{\cdot}{\cdot}$,
that we denote as $\vssemv{\cdot}{\cdot}$, which is defined inductively, but in which the first two rules are
replaced by the two rules
$$
\zrule
    {
      \vssemv{\termone}{\{\termone^1\}}
    }
    {}
$$
$$
\zrule
    {
      \vssemv{\termone}{\emdist}
    }
    {}
$$
This means, in particular, that the relation above maps terms to distribution over terms, rather than
distributions over values.
\begin{lemma}\label{lemma:varvsnormv}
If $\vssemv{\termone}{\distone}$, where $\supp{\distone}\subseteq\val$, then
$\ssemv{\termone}{\distone}$. Conversely, if $\ssemv{\termone}{\distone}$
then $\vssemv{\termone}{\distone}$.
\end{lemma}
\begin{proof}
Simple inductions.
\end{proof}

\begin{lemma}\label{lemma:redKn2}
If $\ssemn{\termone}{\distone}$, then
$\vssemv{\adm{\termone}{\ctone}}{\adm{\distone}{\ctone}}$.
\end{lemma}
\begin{proof}
By induction on the complexity of a derivation for $\ssemn{\termone}{\distone}$:
\begin{varitemize}
\item
  If $\ssemn{\termone}{\emdist}$, then clearly
  $\vssemv{\adm{\termone}{\ctone}}{\adm{\emdist}{\ctone}}$,
  simply because $\adm{\emdist}{\ctone}=\emdist$.
\item
  If $\termone\rn\mul{\termtwo}$ and $\ssemv{\termtwo_i}{\disttwo_i}$,
  where $\distone=\sum_{i=1}^\natone\frac{1}{\natone}\disttwo_i$, then,
  by Lemma~\ref{lemma:auxrednK},
  $\adm{\termone}{\ctone}\rv\termthree_1,\ldots,\termthree_\natone$ and
  $\termthree_i\rn^*\adm{\termtwo_i}{\ctone}$.
  By induction hypothesis,
  $\vssemv{\adm{\termtwo_i}{\ctone}}{\adm{\disttwo_i}{\ctone}}$ and
  so we can easily form a derivation for
  $\vssemv{\adm{\termone}{\ctone}}{\sum_{i=1}^\natone(\frac{1}{\natone}\adm{\disttwo_i}{\ctone})}$
  But clearly, 
  $$
  \sum_{i=1}^\natone\left(\frac{1}{\natone}\adm{\disttwo_i}{\ctone}\right)=
  \adm{\left(\sum_{i=1}^\natone\frac{1}{\natone}\disttwo_i\right)}{\ctone}
  $$
\end{varitemize}
This concludes the proof.
\end{proof}

\begin{lemma}\label{lemma:redKn3}
If $\ssemv{\adm{\termone}{\idt}}{\distone}$, then $\ssemn{\termone}{\disttwo}$, where
$\Phi(\disttwo)\geq\distone$.
\end{lemma}
\begin{proof}
The proof goes by induction on the complexity of the derivation $\derone$ of
$\ssemn{\adm{\termone}{\abstr{\varone}{\varone}}}{\distone}$. Some interesting cases:
\begin{varitemize}
\item If $\termone$ is a value $\valone$:
\begin{varitemize}
\item
  if $\ssemn{\adm{\valone}{\ide}}{\emdist}$, then $\disttwo=\emdist$ itself;
  \item
  if $\ssemn{\adm{\valone}{\ide}}{\distone}$ with $\distone\neq\emdist$,  then $\distone=\{\valone^{1}\}$.
\end{varitemize}
\item If $\termone$ is not a value:
\begin{varitemize}
\item if $\ssemn{\adm{\termone}{\ide}}{\emdist}$ the thesis follows trivially;
\item if $\ssemn{\adm{\termone}{\ide}}{\distone}$ with $\distone\neq\emdist$, let
  $\termtwo_1,\ldots\termtwo_\natone$ be such that
  $\termone\rv\mul{\termtwo}$. By Lemma~\ref{lemma:auxredK},
  there must be derivations $\dertwo_1,\ldots,\dertwo_n$ (all smaller than $\derone$)
  such that $\dertwo_i:\ssemn{\adm{\termtwo_i}{\idt}}{\disttwo_i}$ and
  $\distone=\sum_{i=1}^\natone\frac{1}{\natone}\disttwo_i$. By induction hypothesis,
  $\ssemv{\termtwo_i}{\distthree_i}$ where $\Phi(\distthree_i)\geq\disttwo_i$ and
  so we can form a derivation of $\ssemv{\termone}{\left(\sum_{i=1}^\natone\distthree_i\right)}$.
  But clearly,
  $$
  \Phi(\sum_{i=1}^\natone\distthree_i)=\sum_{i=1}^\natone \Phi(\distthree_i)\geq\sum_{i=1}^\natone\disttwo_i=\distone.
  $$
\end{varitemize}
\end{varitemize}
This concludes the proof.
\end{proof}

\noindent
\begin{proof}[Theorem~\ref{th:simNbyV}]
To prove that $\Phi(\Ssemn{\termone})=\Ssemv{\app{\mapvn{\termone}}{\idt}}$,
we proceed by showing the following:
\begin{varenumerate}
\item\label{point:nvsim1}
  If $\ssemv{\termone}{\distone}$, then $\ssemn{\app{\mapvn{\termone}}{\idt}}{\disttwo}$
  and $\disttwo\geq\Psi(\distone)$.
\item\label{point:nvsim2}
  If $\ssemn{\app{\mapvn{\termone}}{\idt}}{\distone}$, then $\ssemv{\termone}{\disttwo}$
  where $\Phi(\disttwo)\geq\distone$.
\end{varenumerate}
To prove point~\ref{point:nvsim1}, we observe that
if $\ssemv{\termone}{\distone}$, then both
$\mapvn{\termone}\ctone\rn^{*}\termone:\ctone$ (by Lemma~\ref{lemma:redKn1}) and
$\ssemv{\adm{\termone}{\ctone}}{\adm{\distone}{\ctone}}$ (by Lemma~\ref{lemma:redKn2} and
Lemma~\ref{lemma:varvsnorm}). Now, notice that, for every value $\valone$,
$$
\adm{\valone}{(\idt)}=\app{(\idt)}{\Phi(\valone)}\rv \Phi(\valone).
$$
As a consequence, it's clear that
$$
\ssemn{\adm{\mapvn{\termone}}{\idt}}{\distone}.
$$
Point~\ref{point:nvsim2} is nothing more than Lemma~\ref{lemma:redKn3}.
This concludes the proof.
\end{proof}

\section{On the Expressive Power of $\LOP$}\label{sec:exppower}
The lambda calculus $\LOP$ is endowed with a very restricted form of (probabilistic) choice, which
is binary and such that both possible outcomes have probability $1/2$. It is thus natural to ask
oneself whether this is a essential restriction or not. In this section, we show that this
is \emph{not} essential, by proving that the set of \emph{representable} distributions on the natural
numbers, namely those which can be denoted by a term of $\LOP$ equals the set of \emph{computable}
distributions, defined in terms of Turing Machines. In this section, we assume to work with
call-by-value reduction, even if everything could be rephrased in call-by-name, in view
of the results in Section~\ref{sec:cbvbycbn}.

We first of all need to define what a computable distribution (on the natural numbers) is. 
This is based on a notion of approximation for the function assigning probabilities to
natural numbers. 
 
In $\LOP$ we are able to represent (up to approximation) probability functions from a suitable domain 
to the real interval $\RR_{[0,1]}$. The notion of \emph{approximating function} is necessary:

\begin{definition}[Approximating function]
Given any function $\funone:\setone\rightarrow\RR_{[0,1]}$, the \emph{approximating function}
for $\funone$, denoted $\appr{\funone}$ is the function from $\setone\times\NN$ to
$\{0,1\}^*$ which on input $(\elone,\natone)$, returns the binary string of length 
$\natone$ containing the first $\natone$ digits of $\funone(\elone)$ in binary notation.
\end{definition}
We define now the class of computable functions:
\begin{definition}[Computable Distributions]
A distribution $\pdone:\NN\rightarrow\RR_{[0,1]}$ is computable iff 
the function $\appr{\pdone}:\NN\times\NN\rightarrow\{0,1\}^*$ is computable.
\end{definition}
In the rest of the section we will use the following notation:
\begin{notation}
Given a proper distribution $\distone$ that assigns nonzero probability only to
representation of natural numbers, the corresponding probability distribution
over the naturals will be denoted with $\pd{\distone}$.
\end{notation}

The next step consists in understanding which class of distributions on the natural numbers
can be captured by $\LOP$, i.e. is the semantics of a lambda term. 
For this reason, we assume to work with a fixed encoding of the
natural numbers into lambda terms, namely the one usually attributed to Scott~\cite{Wadsworth80}.
\begin{definition}[Representable Distribution]\label{def:repdist}
A probability distribution $\pdone$ over the natural numbers is said to be
\emph{representable} iff there is a lambda term $\termone_\pdone$ such that
$\pd{\Ssemv{\termone}}=\pdone$.
\end{definition}

Our main goal will be to prove that the class of representable distributions will coincide with the class of 
computable distributions: on one hand each distribution obtained by small step evaluation of a lambda term in 
$\LOP$ is proved to be computable (Soundness Theorem~\ref{th:sound}); on the other hand, each computable 
distribution can be represented by a term in $\LOP$  (Completeness Theorem~\ref{th:compl}).

Whatever can be denoted by a lambda term is actually a computable distribution,
i.e. one which can be approximated up to any degree of precision by a Turing Machine:

\begin{theorem}[Soundness]\label{th:sound}
Every representable distribution is computable.
\end{theorem}
\begin{proof}
Suppose $\pdone$ is a representable probability distribution. Then there
is a term $\termone$ such that $\pd{\Ssemv{\termone}}=\pdone$. This implies,
in particular, that $\sumd{\Ssemv{\termone}}=1$, because $\sumd{\pdone}=1$
itself. An algorithm computing $\appr{\pdone}$, then can be easily designed
as an evaluator for $\termone$: on input $(\elone,\natone)$, simply
compute distributions $\distone$ such that
$\ssemv{\termone}{\distone}$, until you find one such that
$\sumd{\distone}$ is big enough as to be able the determine the probability of
all values up to the $\natone$-th bit.
\end{proof}
The completeness theorem requires other definitions and some auxiliary results. 
To prove expressiveness results, we exploit standard lambda calculus encodings. 
Natural numbers and binary strings can be encoded as follows:
\begin{align*}
\enc{0}&=\abstr{\varone\vartwo}{\varone}\\
\enc{\natone+1}&=\abstr{\varone\vartwo}{\app{\vartwo}\enc{\natone}}\\
\enc{\varepsilon}&=\abstr{\varone\vartwo\varthree}{\varone}\\
\enc{0\cdot\strone}&=\abstr{\varone\vartwo\varthree}{\vartwo\enc{\strone}}\\
\enc{1\cdot\strone}&=\abstr{\varone\vartwo\varthree}{\varthree\enc{\strone}}
\end{align*}
Moreover, we encode pairs of values as follows:
\begin{align*}
\pair{\valone}{\valtwo}&=\abstr{\varone}{\app{\varone}{\valone\valtwo}}
\end{align*}
We define the class of the so-called \emph{finite distribution terms}:
\begin{definition}[Finite distribution terms]\label{def:findt}
Finite distribution terms are lambda-terms generated inductively as follows:
\begin{varitemize}
\item
  For every $\natone$, $\abstr{\varone\vartwo}{\varone\enc{\natone}}$ is 
  a finite distribution term;
\item
  If $\termone$ and $\termtwo$ are finite distribution terms, then
  $\abstr{\varone\vartwo}{\app{\app{\vartwo}{\termone}}{\termtwo}}$ is 
  a finite distribution term.
\end{varitemize}
\end{definition}
Given a finite distribution term $\termone$, the underlying probability distribution
on the natural numbers $\pd{\termone}$ is defined in the natural way:
\begin{align*}
\pd{\abstr{\varone\vartwo}{\varone\enc{\natone}}}&=\{\natone^1\}\\
\pd{\abstr{\varone\vartwo}{\app{\app{\vartwo}{\termone}}{\termtwo}}}&=
   \frac{1}{2}\pd{\termone}+\frac{1}{2}\pd{\termtwo}.
\end{align*}
Lemma~\ref{lemma:fix} is an intermediate result towards Lemma~\ref{lemma:fdt}.

\begin{lemma}[Fixed-point Combinator]\label{lemma:fix}
There is a term $\termfpc$ such that for every value $\valone$,
$\app{\termfpc}{\valone}$ rewrites deterministically to
$\app{\valone}{(\abstr{\varone}{\app{\app{\termfpc}{\valone}}{\varone}})}$.
\end{lemma}
\begin{proof}
The term $\termfpc$ is simply $\app{\valtwo}{\valtwo}$ where
$$
\valtwo=\abstr{\varone}{\abstr{\vartwo}{\app{\vartwo}{(\abstr{\varthree}{\app{\app{\app{\varone}{\varone}}{\vartwo}}{\varthree}})}}}.
$$
Indeed:
\begin{align*}
  \app{\termfpc}{\valone}&=\app{\app{\valtwo}{\valtwo}}{\valone}\\
     &\rv\app{(\abstr{\vartwo}{\app{\vartwo}{(\abstr{\varthree}{\app{\app{\app{\valtwo}{\valtwo}}{\vartwo}}{\varthree}})}})}{\valone}\\
     &\rv\app{\valone}{(\abstr{\varthree}{\app{\app{\app{\valtwo}{\valtwo}}{\valone}}{\varthree}})}\\
     &=\app{\valone}{(\abstr{\varthree}{\app{(\app{\termfpc}{\valone})}{\varthree}})}.
\end{align*}
This concludes the proof.
\end{proof}

\begin{lemma}\label{lemma:fdt}
There is a lambda term $\termfdt$ such that for every finite distribution
term $\termtwo$, $\pd{\Ssemv{\app{\termfdt}{\termtwo}}}=\pd{\termtwo}$.
\end{lemma}
\begin{proof}
The term $\termfdt$ is simply $\app{\termfpc}{\valone}$ where
$\valone$ is
$$
\abstr{\varone\vartwo}{\app{\app{\vartwo}{(\abstr{\varthree}{\varthree})}}{(\abstr{\varthree\varfour}{\ps{(\app{\varone}{\varthree})}{(\app{\varone}{\varfour})}})}}
$$
Indeed, by induction ($\valthree$ is $\abstr{\varone}{\app{\app{\termfpc}{\valone}}{\varone}}$):
\begin{align*}
\app{\termfdt}{(\abstr{\varone\vartwo}{\varone\enc{\natone}})}&\rv^*\app{\app{(\abstr{\varone\vartwo}{\varone\enc{\natone}})}{(\abstr{\varthree}{\varthree})}}
   {(\abstr{\varthree\varfour}{\ps{(\app{\valthree}{\varthree})}{(\app{\valthree}{\varfour})}})}\\
   &\rv^{*}\app{(\abstr{\varthree}{\varthree})}{\enc{\natone}}\rv\enc{\natone};\\
\app{\termfdt}{(\abstr{\varone\vartwo}{\app{\app{\vartwo}{\termone}}{\termtwo}})}&\rv^*\app{\valone}{\valthree(\abstr{\varone\vartwo}{\app{\app{\vartwo}{\termone}}{\termtwo}})}\\
&\rv^{*}\app{(\abstr{\varthree\varfour}{\ps{\app{\valthree}{\varthree}}{\app{\valthree}{\varfour}}})}{\app{\termone}{\termtwo}}\\
&\rv^{*}\ps{(\app{\valthree}{\termone})}{(\app{\valthree}{\termtwo})}
\end{align*}
and $\ps{(\app{\valthree}{\termone})}{(\app{\valthree}{\termtwo})}$, applying the i.h. on the subterms,
evaluates with equiprobable distribution to either $\pd{\termone}$, $\pd{\termtwo}$.
This concludes the proof.
\end{proof}

The following proposition guarantees the existence of a sort of successive approximations function, 
which computes the distribution of a given term in a finite number of steps.

\begin{proposition}[Splitting]\label{prop:split}
There is a lambda term $\termsplit$ such that for every term $\termtwo$ 
computing the distribution $\pdone$, $\app{\termsplit}{\termtwo}$ rewrites deterministically
to $\pair{\termthree}{\termfour}$ where 
\begin{varitemize}
\item
  $\termthree$ is a finite distribution term such that
  $\pd{\termthree}=\pdtwo$;
\item
  $\termfour$ computes a distribution $\pdthree$;
\item
  $\pdone=\frac{1}{2}\pdtwo+\frac{1}{2}\pdthree$.
\end{varitemize}
\end{proposition}
\begin{proof}
Only informally, observe that it is possible to find the (sub)distribution $\pdtwo$ in a finite numbers of steps,  
by successively querying $\termtwo$ (which takes two natural numbers as arguments) in the dovetail order. This
way, the (sub)distribution $\pdtwo$ can always be determined.
\end{proof}

By means of Lemma~\ref{lemma:fdt} and Proposition~\ref{prop:split} it is possible to prove the following completeness result:
\begin{theorem}[Completeness]\label{th:compl}
Every computable probability distribution $\pdone$ over the natural numbers is representable by a term
$\termone_\pdone$.
\end{theorem}
\begin{proof}
Define a term $\termone$ simply as $\app{\termfpc}{\valone}$ where
$\valone$ is
$$
\abstr{\varone\vartwo}{\app{(\app{\termsplit}{\vartwo})}{(\abstr{\varthree\varfour}{[\app{{\abstr{\varfive}{}}}{\ps{(\app{\termfdt}{\varthree})}{\abstr{\varfive}{(\app{\varone}{\varfour})}}}]}{(\abstr{\varfive}{\varfive})})}}
$$

We will prove the following statement: for any $\natone\in\NN$ e for all $\termtwo\in\LOP$ which computes a distribution $\disttwo$,  there exists a distribution $\distone$ such that $\ssemv{\app{\termone}{\termtwo}}{\distone}$ with $\sumd{\distone}\geq(1-\frac{1}{2^{\natone}})$ and $\pd{\distone}\leq\disttwo$.
\noindent
We will prove the thesis by induction on the natural number $\natone$.
\begin{varitemize}
\item If $\natone=0$ then $\distone=\emdist$ and $\ssemv{\app{\termone}{\termtwo}}{\emdist}$ trivially.
\item If $\natone\gt0$ observe that $\app{\termone}{\termtwo}$ rewrites as follows:

\begin{align*}
\app{\termone}{\termtwo}&=\app{\termfpc}{\valone\termtwo}\rvn^{*}\app{{\app{\valone}{(\abstr{\varsix}
   {\app{(\app{\termfpc}{\valone})}{\varsix}})}}}{\termtwo}\\
&\rvn^{*}\app{(\app{\termsplit}{\termtwo})}{\underbrace{(\abstr{\varthree\varfour}{((\app{(\ps{(\abstr{\varfive}{\app{\termfdt}{\varthree}})}{(\abstr{\varfive}{\app{(\abstr{\varsix}{\app{\termfpc}{\valone\varsix}})}{\varfour}})})}{(\abstr{\varfive}{\varfive})})}))}}\\
& \hspace{38ex}\termsix.
\end{align*}
Observe that, by Proposition~\ref{prop:split}, $\app{\termsplit}{\termtwo}$ rewrites deterministically to the pair 
$\pair{\termthree}{\termfour}$, where $\termthree$ is a finite distribution term such that $\pd{\termthree}= 
\disttwo_{\termthree}$ and $\termfour $ is a finite distribution term such that $\pd{\termfour}= \disttwo_{\termfour}$ 
and $\disttwo=\frac{1}{2}\disttwo_{\termthree}+\frac{1}{2} \disttwo_{\termfour}$.
Then: 
\begin{align*}
  \app{(\abstr{\varone}{\app{\varone}{\termthree\termfour}})}{\termsix}
  &\rvn \app{(\abstr{\varthree\varfour}{((\app{(\ps{(\abstr{\varfive}{\app{\termfdt}{\varthree}})}{(\abstr{\varfive}{\app{(\abstr{\varsix}{\app{\termfpc}{\valone\varsix}})}{\varfour}})})}{(\abstr{\varfive}{\varfive})})}))}{\app{\termthree\termfour}}\\
  &\rvn^{2}{{\app{(\ps{(\abstr{\varfive}{\app{\termfdt}{\termthree}})}{(\abstr{\varfive}{\app{(\abstr{\varsix}{\app{\termfpc}{\valone\varsix}})}{\termfour}})})}{(\abstr{\varfive}{\varfive})}}}\\
  &\rvn \app{(\abstr{\varfive}{\app{\termfdt}{\termthree}})}{(\abstr{\varfive}{\varfive})},\app{(\abstr{\varfive}{\app{(\abstr{\varsix}{\app{\termfpc}{\valone\varsix}})}{\termfour}})}{(\abstr{\varfive}{\varfive})}.
\end{align*}
Now observe that, $\app{(\abstr{\varfive}{\app{\termfdt}{\termthree}})}{(\abstr{\varfive}{\varfive})}\rvn\app{\termfdt}{\termthree}$ and by Lemma~\ref{lemma:fdt} this term evaluates to $\pd{\termthree}=\disttwo_{\termthree}$.
Moreover, $\app{(\abstr{\varfive}{\app{(\abstr{\varsix}{\app{\termfpc}{\valone\varsix}})}
{\termfour}})}{(\abstr{\varfive}{\varfive})}\rvn^{2}\app{(\termfpc\valone)}{\termfour}=\app{\termone}{\termfour}$ 
and we can apply to $\termfour$ the induction hypothesis. Since $\termfour$ compute a distribution $\disttwo_{\termfour}$ 
by induction hypothesis for each natural number $\natone$, there exist a distribution ${\distone_{\termfour}}$ such
that $\ssemv{\app{\termone}{\termfour}}{{\distone_{\termfour}}}$ with $\sumd{{\distone}_{\termfour}}\geq(1-\frac{1}{2^n})$ 
and $\pd{{\distone_{\termfour}}}\leq\disttwo_{\termfour}$.

Let us take the distribution $\distone$ as $\distone=\frac{1}{2}\distone_{\termthree}+\frac{1}{2}\distone_{\termfour}$.
\noindent
Then:
\begin{align*}
 \sumd{\distone}&=\frac{1}{2}\sumd{\disttwo_{\termthree}}+\frac{1}{2}\sumd{\distone_{\termfour}}\\
 &\stackrel{i.h.}{\geq} \frac{1}{2}+  \frac{1}{2}\cdot(1- \frac{1}{2^n})\\
 & \frac{1}{2}+ \frac{1}{2}- \frac{1}{2^{n+1}}=1- \frac{1}{2^{n+1}}.
 \end{align*}
Moreover, $\distone\leq\disttwo$.
The thesis follows easily.
\end{varitemize}
\end{proof}

\section{Conclusions and Future Work}

In this paper we studied probabilistic operational semantics for $\LOP$, a nondeterministic extension 
of untyped lambda calculus. We prove strong equivalence results between small-step and 
big-step semantics, both in call-by-value and in call-by-name. We also extend 
Plotkin's simulation to our probabilistic setting and we state and 
prove some results about the expressive power of $\LOP$.

Starting from the present paper, several directions for future work are open. On the one hand, 
some theoretical aspects of the calculus remain unexplored: for example, it should be an interesting
topic to develop an observational theory for $\LOP$ terms. Moreover, the equivalence between inductive
and co-inductive semantics seem to be reminiscent of equality between outer
measure and inner measure in measure theory: it should be interesting to prove some results in this direction. 
On the other hand, it is possible to consider $\LOP$  as a paradigmatic language for stochastic functional 
programming: it should be fascinating to analyze carefully the relationship between $\LOP$ and other 
probabilistic languages (for example, Park's language $\parkl$~\cite{Park03}).

\bibliography{biblio.bib} 
\bibliographystyle{abbrv}
\end{document}